\documentclass[11pt]{article}

\usepackage[final]{acl}
\usepackage{times}
\usepackage{latexsym}
\usepackage[T1]{fontenc}
\usepackage[utf8]{inputenc}
\usepackage{microtype}
\usepackage{inconsolata}

\usepackage{graphicx} 
\usepackage{subcaption} 
\usepackage{balance}
\usepackage{amsthm}

\usepackage{amssymb}
\usepackage{indentfirst}
\usepackage{amsmath}
\usepackage{xcolor}

\usepackage{booktabs}
\usepackage{caption}
\usepackage{subcaption}
\usepackage{float}
\usepackage{multirow}
\usepackage{tabularx}
\usepackage{arydshln}
\usepackage{adjustbox}

\usepackage[ruled,longend,linesnumbered]{algorithm2e}

\newcommand{\modelname}{IKE}
\newcommand{\modelnamevd}{IKE$_{VD}$}

\newtheorem{prop}{Proposition}
\newtheorem{thm}{Theorem}
\newtheorem{definition}{Definition}
\newtheorem{col}{Corollary}

\title{LLMs Meet Isolation Kernel: \\ An (Almost) Learning-free Binary Embeddings for Fast Retrieval}
\title{LLMs Meet Isolation Kernel: \\ Lightweight, Learning-free Binary Embeddings for Fast Retrieval}

\author{Zhibo Zhang\textsuperscript{1 2}  \quad  Yang Xu\textsuperscript{1 2}  \quad  Kai Ming Ting\textsuperscript{1 2}  \quad  Cam-Tu Nguyen\textsuperscript{1 2}\thanks{Corresponding author} \\
        \textsuperscript{1}National Key Laboratory for Novel Software Technology, Nanjing University \\
        \textsuperscript{2}School of Artifcial Intelligence, Nanjing University \\
        zhibozhang@smail.nju.edu.cn, xuyang@lamda.nju.edu.cn,
        \{tingkm, ncamtu\}@nju.edu.cn
}
  
\begin{document}
\maketitle
\begin{abstract}

Large language models (LLMs) have recently enabled remarkable progress in text representation. However, their embeddings are typically high-dimensional, leading to substantial storage and retrieval overhead. Although recent approaches such as Matryoshka Representation Learning (MRL) and Contrastive Sparse Representation (CSR) alleviate these issues to some extent, they still suffer from retrieval accuracy degradation.
This paper proposes \emph{Isolation Kernel Embedding} or \modelname{}, a learning-free method that transforms an LLM embedding into a binary embedding using Isolation Kernel (IK). 
Lightweight and based on binary encoding, IKE offers a low memory footprint and fast bitwise computation, lowering retrieval latency. 
Experiments on multiple text retrieval datasets demonstrate that \modelname{} offers up to 16.7\(\times\) faster retrieval and 16\(\times\) lower memory usage than the original LLM embeddings, while maintaining comparable accuracy.
Theoretically, we show that \modelname{} works because it satisfies four essential criteria for effective binary hashing that other methods do not possess.
Compared to CSR, \modelname{} consistently achieves better retrieval efficiency and effectiveness.
IKE also works effectively with graph-based indexing, demonstrating its superiority in balancing accuracy and latency compared to alternative compression techniques in the approximate nearest neighbor (ANN) search setting.

\end{abstract}

\section{Introduction}\label{sec:inrto}

In the last decade, pretrained language models have significantly advanced text representation. Early approaches relied on encoder-based models such as BERT \cite{bert} and RoBERTa \cite{roberta}. Recently, there has been a paradigm shift toward leveraging large language models (LLMs) for text representation with notable examples including LLM2Vec \cite{llm2vec}, Llama2Vec \cite{llama2vec}, and Qwen Embedding \cite{qwen3embedding}. These LLM embeddings now represent the state of the art in text representation \cite{enevoldsen2025mmtebmassivemultilingualtext}.

However, LLM embeddings typically have a high number of dimensions. For example, T5-11B \cite{t5-11b} has 1024 dimensions, while the mainstream 7B decoder-only LLMs use 4096 dimensions. Such high-dimensional embeddings substantially increase memory storage and retrieval latency. To address this, \citet{kusupati2022matryoshka} introduced Matryoshka Representation Learning (MRL) that allows adaptive embedding length, balancing latency and accuracy for different applications. This approach is now supported in OpenAI and Google Gemini embedded APIs \cite{openai_mrl, lee2024gecko_mrl} and has been extended to various applications \cite{openai2024embedding,nussbaum2025nomic,yu2024arctic}. 
However, MRL requires full-parameter tuning, which results in a high training cost. Furthermore, experiments in \cite{wenbeyond} reveal a marked degradation in retrieval performance for MRL. 


Recently, \citet{wenbeyond} proposed Contrastive Sparse Representation (CSR) that sparsifies a pretrained LLM embedding into a high-dimensional but selectively activated feature space. Compared to MRL, CSR freezes the LLM backbone and trains only a lightweight adapter (an MLP layer), significantly reducing the training cost. However, unlike MRL, CSR requires retraining for different code lengths and suffers from substantial transformation overhead (see Section \ref{sec:learningBased}).

In a nutshell, MRL and CSR fall short of achieving the aim of fast retrieval time with low memory requirement, without substantially degrading the retrieval performance on a high-dimensional LLM embedded database.







\begin{figure*}[htbp]
  \centering

\begin{subfigure}{\textwidth}
    \centering 
    \includegraphics[width=0.6\linewidth]{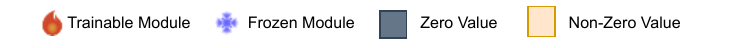}
  \end{subfigure}
  
  \begin{subfigure}{0.3\textwidth}
    \parbox[][4cm][c]{\linewidth}{ 
    \centering 
    \includegraphics[width=\linewidth]{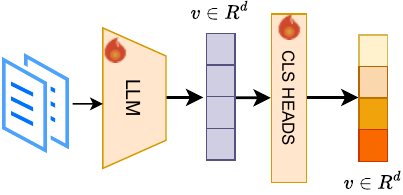}
    }
    \caption{MRL (Learning-based)}
    \label{fig:motivation_MRL}
  \end{subfigure}\hfill
  \begin{subfigure}{0.3\textwidth}
    \parbox[][4cm][c]{\linewidth}{ 
    \centering 
    \includegraphics[width=\linewidth]{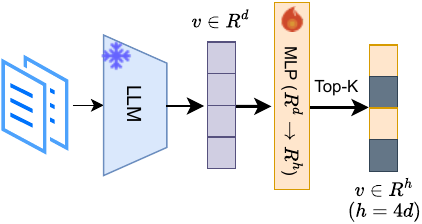}
    }
    \caption{CSR (Learning-based)}
    \label{fig:motivation_CSR}
  \end{subfigure}\hfill
  \begin{subfigure}{0.35\textwidth}
    \parbox[][4cm][c]{\linewidth}{ 
    \centering 
    \includegraphics[width=\linewidth]{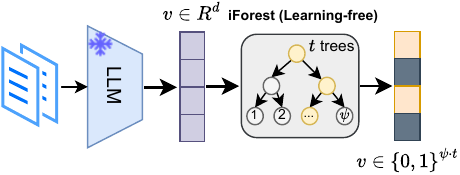}
    }
    \caption{IKE (No Learning)}
    \label{fig:motivation_IKE}
  \end{subfigure}
  \caption{MRL learns adaptive embedding length by optimizing different downstream applications. CSR combines reconstruction loss and contrastive loss for learning a sparse representation. In contrast, IKE maps points into binary embeddings for efficient retrieval through random partitions without sophisticated learning.}
  \label{fig:motivation_example}
  
\end{figure*}

While MRL and CSR provide a solution to the above retrieval problem, the methods are based on ad-hoc learning objectives, created without understanding the criteria of a feature space that are necessary in addressing the problem. A recent work has articulated three key criteria \cite{xu2025VDeH}: full space coverage, entropy maximization, and bit independence. The hashing method called VDeH converts a high-dimensional space in $\mathbb{R}^d$ to a binary embedded space, satisfying the three criteria without learning.

Though VDeH \cite{xu2025VDeH} has been shown to outperform many learning-to-hash (L2H) methods \cite{rcLSH,SH,CBE,ITQ,BP,SP,DSH,SELVE,SpH, li2018L2H} for non-LLM-related databases, our early examination has revealed that VDeH performs poorly in LLM-embedded databases (see Figure \ref{fig:para_m}). This paper reports our effort in creating a generalized version of VDeH that is amenable to LLM-embedded databases.

Our contributions are summarized as follows:
\begin{itemize}

    \item Proposing a novel feature transformation method named \modelname{} that efficiently and effectively converts an LLM embedding to a binary embedding without learning. IKE is a generalization of VDeH, which produces an embedding that satisfies four criteria. The fourth criterion, which we have identified, is that the partitions used to generate the embedding shall have \emph{high diversity}.   
    \item Conducting a theoretical analysis that not only uncovers the \emph{high diversity} criterion, but also the relative diversity of variants of two different implementations of IKE, i.e., Voronoi Diagrams and Isolation Forest.
    \item We conducted comprehensive experiments, where the main findings are: 1) \textit{In the exhaustive search}, IKE generally achieves 2.5-16.7x speedup in search time and 8-16x storage reduction, while maintaining performance of MRR@10 between 98\% and 101\% of the original LLM space; 2) \textit{In ANN (Approximate Nearest Neighbor Search)}, pairing IKE with an HNSW index achieves up to 5x higher throughput than operating in LLM space; 3) \textit{Further comparison} with learning-based method (CSR) and other (no-learning) compression methods (hashing and quantization) demonstrate the advantages of IKE.
\end{itemize}

\section{Preliminaries}\label{sec:preliminaries}
\paragraph{Isolation Kernel}
The Isolation Kernel~\cite{ting2018isolation} (IK) is a data-dependent kernel that computes the similarity between two points based on the partitions in the data space. Specifically, given a data set $D=\{\mathbf{x}_1,\dots,\mathbf{x}_n\} \subset \mathbb{R}^d$, we define $\mathcal{D} \subset D$ as a subset of $\psi$ data points, each drawn from $D$ with equal probability $\frac{1}{n}$. A \textit{partition} $H$ is derived so that each \textit{isolating partition} $\theta \in H$ isolates a point from the rest of  $\mathcal{D}$. 

Let $\mathcal{H}_\psi(D)$ denote the set of all partitions that are admissible under the data set $D$. IK of any two points $\mathbf{x},\mathbf{y}\in \mathbb{R}^d$ wrt $D$ is defined to be the expectation taken over the probability distribution $P_{\mathcal{H}_\psi(D)}$ on all partitions $H \in \mathcal{H}_\psi(D)$ that both $\mathbf{x}$ and $\mathbf{y}$ fall into the same isolating partition $\theta \in H$:
\begin{equation}
    \begin{aligned}
    K_{\psi,D}(\mathbf{x},\mathbf{y})=\mathbb{E}_{H \sim P_{\mathcal{H}_\psi(D)}}[\mathbb{I}(\mathbf{x},\mathbf{y}\in\theta \mid \theta\in H)],\nonumber
    \end{aligned}
\end{equation}
where $\mathbb{I}(\cdot)$ is an indicator function.

In practice, Isolation Kernel $K_\psi$ is estimated from a finite number of $t$ partitions $H_i$, each is derived from a data subset $\mathcal{D}_i\subset D$:
\begin{equation}
    K_{\psi,D}(\mathbf{x},\mathbf{y}) 
    \simeq \frac{1}{t} \, \sum_{i=1}^t \sum_{\theta_j \in H_i} \mathbb{I}(\mathbf{x} \in \theta_j)\mathbb{I}(\mathbf{y} \in \theta_j).\nonumber
\end{equation}
\paragraph{Voronoi Diagram Encoded Hashing} In the preliminary study, Voronoi Diagrams (VD) are explored as an implementation of the Isolation Kernel for transforming LLM embeddings into binary codes. This approach, known as VDeH \cite{xu2025VDeH}, has been shown to exhibit three desirable criteria: full space coverage, maximum entropy, and bit independence. However, VDeH suffers from notable performance degradation compared to the original LLM space (see Figure \ref{fig:para_m}). 

We hypothesize that the performance drop of VDeH is because the relationship between samples (points) in the LLM embedding space is more complex than what VDeH captures. In particular, using all dimensions for isolating partitions (i.e., hyperplanes in the VD) restricts the diversity of possible partitions. To verify this, we propose a generalization of VDeH, termed \modelnamevd{}, where for each partition we randomly sample m (m < d) dimensions for building VD. Note that when $m=d$, \modelnamevd{} becomes VDeH. The randomization of choosing isolation dimensions increases the diversity of isolating partitions, which becomes particularly beneficial when the number of partitions ($t$) is large. As a result, \modelnamevd{} achieves superior performance with small values of (m) compared to VDeH, as shown in Figure \ref{fig:para_m}. 

        

        


\begin{figure}[]
  \centering
    \includegraphics[width=0.9\columnwidth, height=3.8cm]{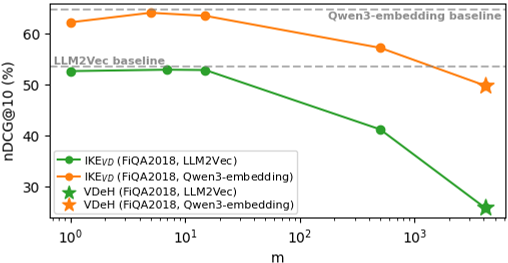}
    \caption{Effect of parameter \(m\) in {\modelnamevd} (\(t\)=4096) on nDCG@10 across two LLM embedding datasets. The pentagram marker indicates the performance of {\modelnamevd} when \(m=d\), which is equivalent to the VDeH method.}
    \label{fig:para_m}
\end{figure}
\section{Isolation Kernel-based Embeddings} \label{sec:IKE}
Building on the previous observation, we investigate \modelname{} using both Isolation Forest (iForest) and VD. The construction of iForest or VD requires no learning, which means neither sophisticated optimization nor even simple greedy search. The iForest or VD captures the data distribution by creating large partitions in sparse regions and small partitions in dense regions. This is created in a completely random process. 

\begin{figure}[]
  \centering
  \begin{subfigure}{\columnwidth}
    \centering
    \includegraphics[width=\linewidth]{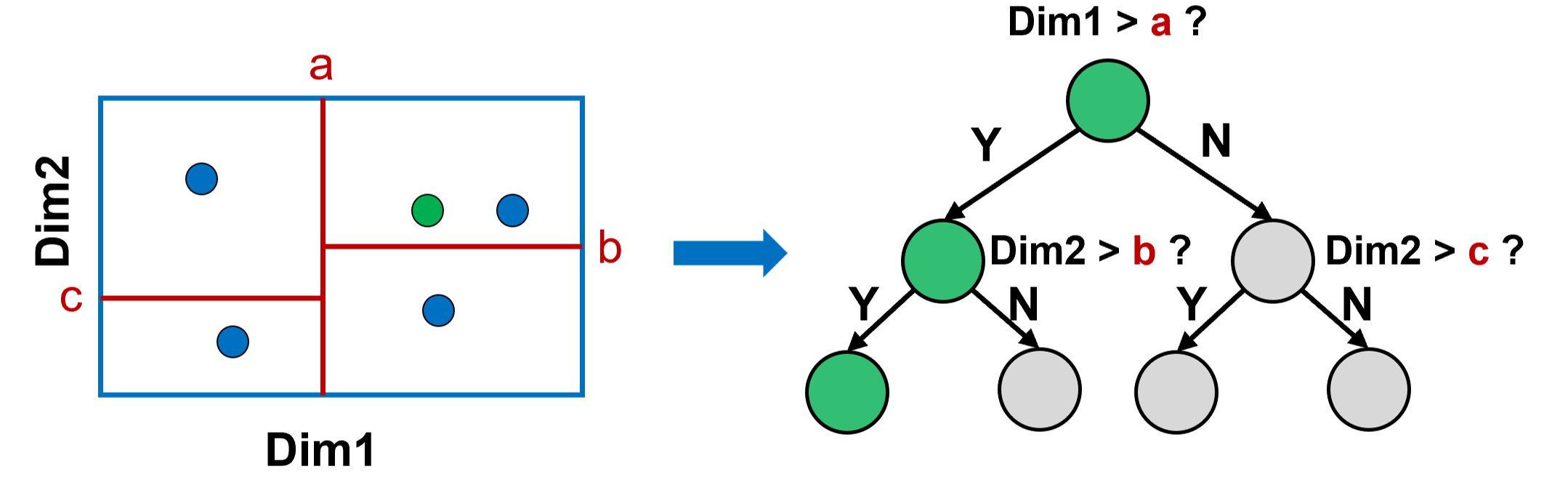}
    \caption{iForest: a, b, c are splitting values for dimensions Dim1 and Dim2 to create the four partitions.}
    \label{fig:iforest_partition}
  \end{subfigure}
  \begin{subfigure}{\columnwidth}
    \centering
    \includegraphics[width=\linewidth]{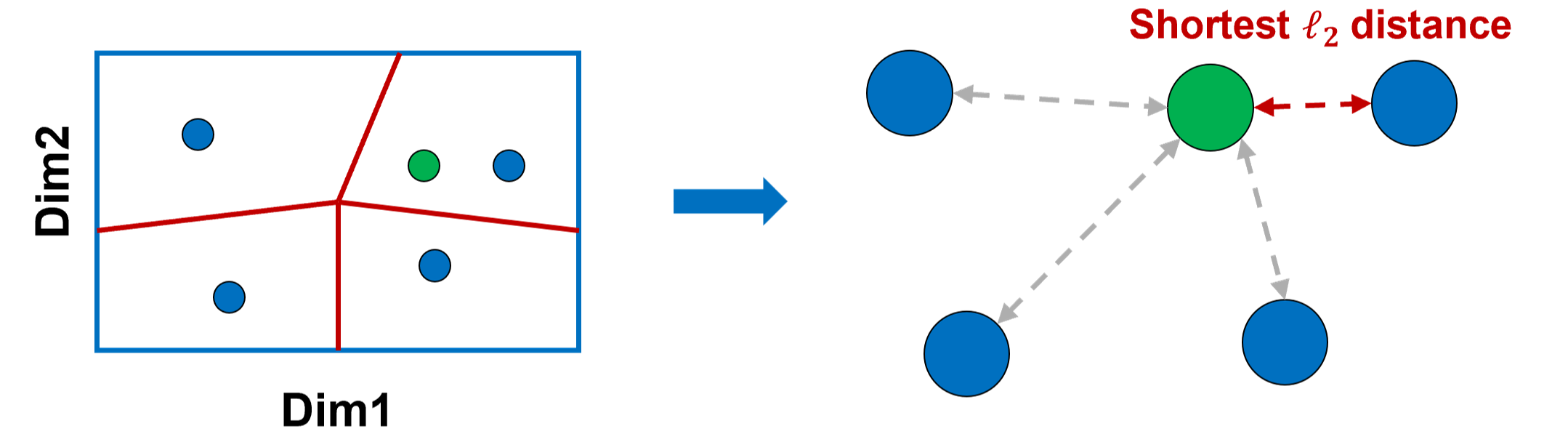}
    \caption{Voronoi Diagram (VD): red lines indicate splitting hyperplanes, which split the space into four VD cells.}
    \label{fig:VD_partition}
  \end{subfigure}
  
  \caption{A illustration of partitioning in 2D space ($\psi=4$) by different IK implementations: iForest and VD. Blue points indicate the samples used to construct the partitions, and the green point represents an input point to be transformed.}
  \label{fig:iforest_and_VD_partition}
\end{figure}

\subsection{IKE based Vector Database Construction}

\paragraph{Embedding} An LLM-based text embedding model is utilized to convert corpus documents and queries into dense vector embeddings. This embedding model is pretrained to effectively capture textual semantics.

\paragraph{IKE with iForest} A random subset of the corpus embeddings is selected to construct an Isolation Forest (iForest) model, where each partition \(H\) is referred to as an \emph{Isolation Tree} (iTree). Consequently, an Isolation Forest represents an ensemble of \(t\) independently constructed partitions.


An iTree \(H\) is a proper binary tree in which every internal node has two children. Each child node corresponds to an isolating partition \(\theta \in H\). Here, a child node is obtained by \textit{randomly selecting a feature \(q\) and the split value \(p\)} uniformly from the interval between the minimum and maximum values of feature \(q\) within \(\mathcal{D}\).
The construction details of the iTree are provided in Appendix \ref{app:method_details:iTree} and the partitioning is illustrated in Figure \ref{fig:iforest_partition}.   

Under the assumption that all points in \(\mathcal{D}\) are distinct, each point will eventually be isolated in its own leaf node. Therefore, an iTree consists of exactly \(\psi\) leaf nodes and \(\psi - 1\) internal nodes, resulting in a total of \(2\psi - 1\) nodes. Hence, the memory complexity of a single iTree is \(O(\psi)\). As previously stated, $t$ iTrees constitute the Isolation Forest; thus, the overall space complexity of the model is \(O(t\psi)\). 

\paragraph{IKE with Voronoi Diagram (\modelnamevd)} Although iForest is used in \modelname{} by default, we can also exploit Voronoi Diagram for IKE as described in the previous section. Note that we refer \modelname{} to indicate the iForest implementation, and \modelnamevd{} for the variant with VD implementation. A formal description of \modelnamevd{} is provided in Appendix \ref{app:method_details:ike_vd} and the partitioning of VD is illustrated in Figure \ref{fig:VD_partition}. An \modelnamevd{} model is composed of $t$ independently constructed partitions, thus exhibiting an overall space complexity of \(O(mt\psi)\).

\paragraph{Vector Database}
Depending on the retrieval strategy, the data points can either be stored directly, which corresponds to an exhaustive search during retrieval, or organized using an indexing structure to enable Approximate Nearest Neighbor (ANN) search. Each data point in IKE space is stored using an index-based representation \(\Phi^{\mathrm{idx}}(\mathbf{x})\), which comprises \(t\) elements. Here, each element \(\Phi^{\mathrm{idx}}_i(\mathbf{x}) \in \{1,\cdots,\psi\}\) indicates the index of the isolating partition to which data point \(\mathbf{x}\) is assigned within the \(i\)-th partition. As each element of \(\Phi^{\mathrm{idx}}(\mathbf{x})\) can be stored using \(\lceil \log_2(\psi) \rceil\) bits, we only need \(t\lceil \log_2(\psi) \rceil\) bits per point. In contrast, we need a 32-bit floating-point number for each dimension in the original LLM embeddings. Therefore, when $t = d$, the transformed IKE yields a \(\frac{32}{\lceil \log_2(\psi) \rceil}\) times reduction in storage usage. 

\paragraph{Adaptability of Embedding Length}  It is desirable to have embeddings that can be truncated for different accuracy–speed tradeoffs as in MRL \cite{kusupati2022matryoshka}. {\modelname} achieves this through its compositional design: since the $t$ iTrees are built independently, any subset can be retained to reduce the transformed space dimension. Unlike MRL, we achieve adaptability without LLM training.



\subsection{Retrieval within the {\modelname} space}

\paragraph{Retrieval.} 
Given a query \(q\) obtained from an LLM-based embedding model, we first apply the previously constructed iForest or VD model to map the query into the {\modelname} space, yielding the transformed vector \(\Phi^{\mathrm{idx}}(\mathbf{q})\). Next, we measure the similarity between \(\Phi^{\mathrm{idx}}(\mathbf{q})\) and each transformed representation \(\Phi^{\mathrm{idx}}(\mathbf{x}) (\forall x\in D)\) for retrieval.
{\begin{equation} \label{eq:Kpsi_idx}
\begin{aligned}
    K_\psi\bigl(\mathbf{x},\mathbf{y} \mid D\bigr)
    \simeq
    \frac{1}{t}\,\sum_{i=1}^t\mathbb{I}(\Phi^{\mathrm{idx}}_i(\mathbf{x})=\Phi^{\mathrm{idx}}_i(\mathbf{y})). \nonumber
\end{aligned}
\end{equation}}


\paragraph{Mapping Time} The mapping time of \modelname{} (with iForest by default) is \(O(t \lceil \log_2(\psi)\rceil)\) as the height limit of each iTree is set to \( \lceil \log_2(\psi) \rceil\). On the other hand, the mapping process of {\modelnamevd} involves \(\ell_2\) distance calculations over \(m\) dimensional vectors, resulting in a time complexity of \(O(mt\psi)\). As a result, \modelname{} is more suitable for online retrieval compared to \modelnamevd{}.

\paragraph{Efficient Similarity Computation}
The similarity measure requires packing \(n_b = \lceil \log_2(\psi) \rceil\) bits into a single unit for comparison. When \(\psi=2\), we have \( n_b = 1 \), in which case the binary comparison can be applied. When \(n_b \geq 2\), we employ the following approach for similarity computation:

Given two IKEs \(x\) and \(y\) to be compared, their partial components are first concatenated into longer bit strings (e.g., 64 bits), denoted as \(D_x\) and \(D_y\). Let $D=D_x \oplus D_y$, where \(\oplus\) denotes the bitwise XOR operation, we compute: 
\[
M = D \; | \;(D \gg 1) \; | \; \cdots \; | \;(D \gg (n_b-1)) 
\]
Here, \(|\) denotes the bitwise OR operation, and \(\gg\) denotes the logical right shift operation. The resulting bit string \(M\) has the following property: within every segment of \(n_b\) bits, if the rightmost bit is 0, it indicates that \(x\) and \(y\) are equal in the corresponding partition. 

To isolate these zero bits and set all other bits to 1 for efficient counting, a mask is applied. The mask is constructed based on the value of \(n_b\) (for instance, when \(n_b=2\), the mask is 10101010...):
\begin{align*}
M &= M \mid mask \\
n_1 &= \operatorname{popcnt}(M)
\end{align*}
After the OR operation with the mask, the number of remaining 0 bits in \(M\) corresponds to the number of elements where \(D_x\) and \(D_y\) match. The function \(\text{popcnt}(M)\) counts the number of 1-bits in \(M\), denoted as \(n_1\). The number of matching elements can then be derived by a simple subtraction. In practice, we only support \(n_b 
\in \{1, 2, 4, 8\}\) bits to keep the codes byte‑aligned. An example is provided in Appendix \ref{app:method_details:example}.

\subsection{Criteria of Isolation Kernel Embedding}
From a hashing perspective, each Voronoi cell or iTree leaf node can be viewed as a hash function that maps a data point to a binary value, where the value of ``1'' indicates that the hash function \textit{covers} the point. Under some assumptions, {\modelname} satisfies the three criteria demonstrated for VDeH~\cite{xu2025VDeH}: 1) \textbf{Full space coverage}: All hash functions cover entire space $\mathbf{x} \in R^d$, or there exists a hash function $g(\mathbf{x})$ that maps $\mathbf{x}$ to the value of $1$; 2) \textbf{Entropy maximization}: For a given dataset, every hash function shall cover approximately the same number of points; 3) \textbf{Bit independence}: All bits in the hash code are mutually independent. The proof is given in Appendix~\ref{app:ike-theory}.

\paragraph{Diversity Condition} By framing both VDeH and \modelname{} as ensemble methods, we introduce \textbf{diverse partitioning} as a fourth, crucial criterion. We prove that \modelname{} generates more robust hash representations than VDeH, owing to its ability to produce a more diverse set of base partitioners. The detailed proof is provided in Appendix~\ref{app:diverse_partitioning}.

\section{Experiments} \label{sec:Experiments}


        
        
        

\begin{table*}[]
    \centering
    \caption{Results with two LLMs in the exhaustive search setting. The ``Space'' column reports the reduction in memory usage achieved by mapping the data into the {\modelname} space, expressed as a multiplicative factor compared to the LLM embedding. In the ``Other'' column, we record auxiliary processing time: for {\modelname}, this includes {\modelname} model construction and feature mapping for all corpus embeddings; for LLM, this includes \(\ell_2\)-normalization of all corpus embeddings and Faiss's exhaustive search index construction. The ``Search'' column measures the search execution time for all queries, with the speedup ratio of our method indicated in parentheses. 
    }
    \label{tab:comp_with_llm}
        \begin{minipage}{0.48\textwidth}
            \centering
            \caption*{LLM2Vec}
            \resizebox{\textwidth}{!}{%
            \begin{tabular}{ccccccc}
                \toprule
                \multirow{2}{*}{\textbf{Dataset}} & \multirow{2}{*}{\textbf{\shortstack{Space\\Type}}} & \multirow{2}{*}{\textbf{\shortstack{Space}}}  & \multicolumn{2}{c}{\textbf{Time (s)}} & \multirow{2}{*}{\textbf{MRR@10}} & \multirow{2}{*}{\textbf{nDCG@10}} \\
                \cmidrule(lr){4-5}
                & & & \textbf{Other} & \textbf{Search} & &  \\ 
                \midrule
                
                \multirow{2}{*}{HotpotQA} & LLM & & 46.33 & 326.61 & \textbf{83.83} & \textbf{74.09}  \\ 
                & {\modelname} ($\psi$=12) & $\downarrow$8\(\times\) & \textbf{18.17} & \textbf{54.82 ($\downarrow$6.0\(\times\))} & 83.69 \(\pm\) 0.13 & 73.24 \(\pm\) 0.09 \\
                
                \cmidrule(lr){1-7}
                \multirow{2}{*}{FiQA2018} & LLM & & 0.29 & 0.47 & \textbf{60.33} & \textbf{53.49} \\
                & {\modelname} ($\psi$=6) & $\downarrow$8\(\times\) & \textbf{0.20} & \textbf{0.063 ($\downarrow$7.5\(\times\))} & 59.97 \(\pm\) 0.32 & 52.59 \(\pm\) 0.31 \\
                
                \cmidrule(lr){1-7}
                \multirow{2}{*}{FEVER-HN} & LLM & & 0.84 & 1.10 &\textbf{91.52} & \textbf{90.45} \\
                & {\modelname} ($\psi$=15) & $\downarrow$8\(\times\) & \textbf{0.59} & \textbf{0.26 ($\downarrow$4.2\(\times\))} & 91.36 \(\pm\) 0.23 & 90.27 \(\pm\) 0.16 \\

                \cmidrule(lr){1-7}
                \multirow{2}{*}{Istella22} & LLM & & 5.05 & 1.56  & 68.11 & \textbf{63.88}   \\ 
                & {\modelname} ($\psi$=16) & $\downarrow$8\(\times\) & \textbf{3.26} & \textbf{0.62 ($\downarrow$2.5\(\times\))} & \textbf{68.92} \(\pm\) 0.88 & 63.86 \(\pm\) 0.50 \\
                
                \cmidrule(lr){1-7}
                \multirow{2}{*}{TREC DL 23} & LLM & & 5.11 & 0.90 & 97.56 & \textbf{73.10}  \\ 
                & {\modelname} ($\psi$=3) & $\downarrow$16\(\times\) & \textbf{2.23} & \textbf{0.16 ($\downarrow$5.6\(\times\))} & \textbf{98.11} \(\pm\) 1.02 & 71.78 \(\pm\) 0.62 \\
                
                \cmidrule(lr){1-7}
                \multirow{2}{*}{Touche2020.V3} & LLM & & 1.57 & 1.44 & \textbf{82.65} & \textbf{51.61} \\
                & {\modelname} ($\psi$=7) & $\downarrow$8\(\times\) & \textbf{1.06} & \textbf{0.095 ($\downarrow$15.2\(\times\))}  & 82.37 \(\pm\) 3.33 & 51.02 \(\pm\) 1.18 \\
                
                \bottomrule
            \end{tabular}
            }
            \label{subtable1}
        \end{minipage}
        \hfill
        \begin{minipage}{0.48\textwidth}
            \centering
            \caption*{Qwen3-embedding}
            \resizebox{\textwidth}{!}{%
            \begin{tabular}{ccccccc}
                \toprule
                \multirow{2}{*}{\textbf{Dataset}} & \multirow{2}{*}{\textbf{\shortstack{Space\\Type}}} & \multirow{2}{*}{\textbf{\shortstack{Space}}}  & \multicolumn{2}{c}{\textbf{Time (s)}} & \multirow{2}{*}{\textbf{MRR@10}} & \multirow{2}{*}{\textbf{nDCG@10}} \\
                \cmidrule(lr){4-5}
                & & & \textbf{Other} & \textbf{Search} & &  \\ 
                \midrule
                
                \multirow{2}{*}{HotpotQA} & LLM & & 41.70 & 332.66 & \textbf{89.76} & \textbf{76.85} \\ 
                & {\modelname} ($\psi$=10) & $\downarrow$8\(\times\) &\textbf{14.51} & \textbf{57.21 ($\downarrow$5.8\(\times\))} & 88.92 \(\pm\) 0.13 & 75.41 \(\pm\) 0.10 \\
                
                \cmidrule(lr){1-7}
                \multirow{2}{*}{FiQA2018} & LLM & & 0.28 & 0.46 & \textbf{72.05} & \textbf{64.66} \\
                & {\modelname} ($\psi$=15) & $\downarrow$8\(\times\) & \textbf{0.24} & \textbf{0.069 ($\downarrow$6.7\(\times\))} & 70.95 \(\pm\) 0.37 & 63.17 \(\pm\) 0.31 \\
                
                \cmidrule(lr){1-7}
                \multirow{2}{*}{FEVER-HN} & LLM & & 0.88 & 1.22 & \textbf{93.91} & \textbf{92.50} \\
                & {\modelname} ($\psi$=13) & $\downarrow$8\(\times\) & \textbf{0.39} & \textbf{0.29 ($\downarrow$4.2\(\times\))} & 93.59 \(\pm\) 0.22 & 92.09 \(\pm\) 0.19 \\

                \cmidrule(lr){1-7}
                \multirow{2}{*}{Istella22} & LLM & & 5.23 & 1.62 & \textbf{65.25} & 59.44 \\ 
                & {\modelname} ($\psi$=11) & $\downarrow$8\(\times\) & \textbf{1.64} & \textbf{0.61 ($\downarrow$2.7\(\times\))} & 65.11 \(\pm\) 0.70 & \textbf{59.63} \(\pm\) 0.55 \\
                
                \cmidrule(lr){1-7}
                \multirow{2}{*}{TREC DL 23} & LLM & & 5.08 & 0.84 & 95.73 & \textbf{73.79}  \\ 
                & {\modelname} ($\psi$=6) & $\downarrow$8\(\times\) & \textbf{2.91} & \textbf{0.32 ($\downarrow$2.6\(\times\))} & \textbf{96.93} \(\pm\) 0.42 & 72.78 \(\pm\) 0.28\\
                
                \cmidrule(lr){1-7}
                \multirow{2}{*}{Touche2020.V3} & LLM & & 1.55 & 1.59 & \textbf{96.94} & \textbf{76.41} \\
                & {\modelname} ($\psi$=16) & $\downarrow$8\(\times\) & \textbf{0.47} & \textbf{0.095 ($\downarrow$16.7\(\times\))} & 95.83 \(\pm\) 1.41 & 73.35 \(\pm\) 1.13 \\
                
                \bottomrule
            \end{tabular}
            }
            \label{subtable2}
        \end{minipage}
\end{table*}

We conduct comprehensive experiments to evaluate the performance of retrieval in the {\modelname} space. For evaluation, we compute MRR@10~\cite{MRR} and nDCG@10~\cite{nDCG} as effectiveness metrics, and report search time and memory consumption as efficiency metrics. Here, MRR@10 computes the mean reciprocal rank of the first relevant document across all queries, whereas nDCG@10 measures the ranking quality by accounting for the position of relevant documents. All experiments were conducted on a Linux server equipped with two Intel Xeon Gold 6330 CPUs (2.00 GHz), providing a total of 56 physical cores (112 logical threads), and 503 GiB of RAM. Both index construction and search operations were performed in-memory using all available threads. Reported search times represent the average of 10 consecutive runs. We conducted the experiment in Section \ref{sec:comparsion_llm} using 10 different random seeds to demonstrate the stability of \modelname{}. For other experiments, we used a single fixed seed.


\subsection{Setup}
\paragraph{Datasets for Evaluation}
We evaluate our method on four text retrieval datasets from the Massive Text Embedding Benchmark (MTEB)~\cite{muennighoff2022mteb} (HotpotQA~\cite{hotpotqa}, FiQA2018~\cite{thakur2021beir}, FEVER-HN~\cite{fever}, Touche2020.V3~\cite{touchev3}), and two multi-level retrieval datasets: Istella22~\cite{istella22} and TREC DL 23~\cite{trecdl23} (with a 1 million subset extracted from their original corpora). More dataset details are provided in Appendix~\ref{app:exp_detail:dataset}. 

\paragraph{Hyperparameter Settings} 
Following the original Isolation Forest framework~\cite{Isolation_Forest}, we set the maximum iTree height to \(l = \lceil \log_2(\psi)\rceil\). Unless otherwise specified, the number of iTrees \(t\) equals the embedding dimension \(d\). The parameter \(\psi\) is tuned per dataset via a single grid search under exhaustive search. The detailed tuning procedure is provided in Appendix~\ref{app:exp_detail:hyperpara}.

\paragraph{LLMs}
We employ two high‐dimensional dense retrieval models: LLM2Vec (Mistral‑7B backbone)~\cite{llm2vec} and Qwen3‑embedding‑8B~\cite{qwen3embedding}, both producing 4096‑dimensional embeddings.

\subsection{Comparison with LLM embeddings}\label{sec:comparsion_llm}

\paragraph{Exhaustive Search}
Table \ref{tab:comp_with_llm} presents the retrieval performance in the {\modelname} space compared with the LLM embedding space under an exhaustive search configuration. The experiment in the LLM embedding space was implemented using the Faiss library \cite{douze2024faiss}. The experimental results demonstrate that retrieval in the {\modelname} space achieves a 2.5–16.7\(\times\) speedup in search time and reduces memory usage by 8–16\(\times\) compared to the LLM embedding space, while maintaining comparable or superior retrieval performance. Specifically, MRR@10 reaches between 98\% and 101\% of those obtained in the LLM embedding space, while nDCG@10 varies from 96\% to 100\%.

\paragraph{ANN Search}
Unlike exhaustive search, ANN methods build an index over the data \cite{grpah-based-ann-survey, ann-multi-dim-database, wei2025subspace, starling, lin-2025-operational} to restrict the search to a small subset of candidates, thereby reducing query time at the cost of a slight accuracy loss. We integrate {\modelname} with two widely used indexing structures: the inverted file index (IVF) and the graph-based HNSW index ~\cite{hnsw}. Their retrieval performance is then compared against counterparts in the LLM embedding space with Faiss implementation (See Appendix \ref{app:exp:llm_emb:ann_details}).




\sloppy
Figure \ref{fig:ann_flat} illustrates the Query per Second (QPS) versus MRR@10 trade-off curves of four ANN methods on the HotpotQA dataset, operated in either the LLM embedding or the IKE spaces. Overall, {\modelname} consistently outperforms that of LLM across different indexing methods. When using the HNSW index, retrieval with {\modelname} achieves approximately 4–5 times higher throughput than with LLM embeddings. Specifically, on the HotpotQA (LLM2Vec) dataset, at a comparable MRR@10 of around 61\%, the QPS increases from \(1.6 \times 10^4\) (HNSW+LLM) to \(6.4 \times 10^4\) (HNSW+{\modelname}). When using the IVF index, performing retrieval in the {\modelname} space yields a QPS improvement of approximately 6–8 times over that in the LLM-embedded space. Additional results on more datasets are provided in Appendix \ref{app:exp:llm_emb:more_exp_res}.

\begin{figure}[]
  \centering
  \begin{minipage}{\linewidth}
    \centering
    
    \subfloat[][HotpotQA (LLM2Vec)]{\label{fig:flat_ann_hotpotqa_l2v}\includegraphics[width=.50\columnwidth]{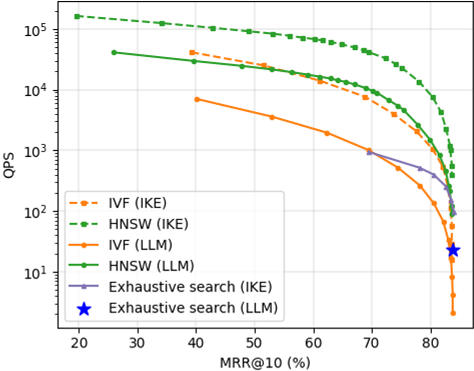}}\hfill
    \subfloat[][HotpotQA (Qwen3)]{\label{fig:flat_ann_hotpotqa_qwen3}\includegraphics[width=.48\columnwidth]{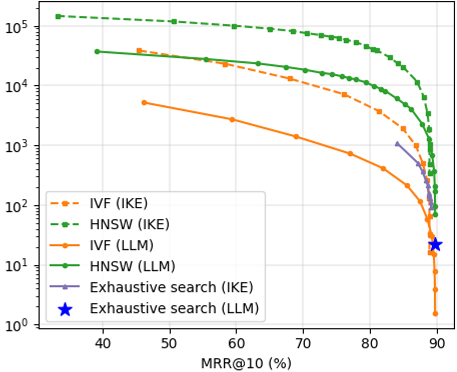}}

  \end{minipage}
  
  \caption{QPS vs. MRR@10 for Four ANN Methods on the HotpotQA dataset. The curves are obtained by adjusting the ANN search parameters. The detailed configuration is described in Appendix \ref{app:exp:llm_emb:config}.} 
  \label{fig:ann_flat}
\end{figure}




\subsection{Comparison with learning-based method}\label{sec:learningBased}

In this section, we compare {\modelname} with compression methods that require learning, such as MRL \cite{kusupati2022matryoshka} and CSR \cite{wenbeyond}. However, as MRL requires full model retraining (see Figure \ref{fig:motivation_example}), we do not have enough computing resources for such an experiment. In addition, MRL is shown to be inferior to CSR, and thus, we select CSR as a representative baseline for comparison.

  

\begin{figure}[]
  \centering
    \begin{minipage}{\linewidth}
    \centering
    
    \subfloat[][Accuracy]{\label{fig:CSR_IKE_acc}\includegraphics[width=.49\columnwidth]{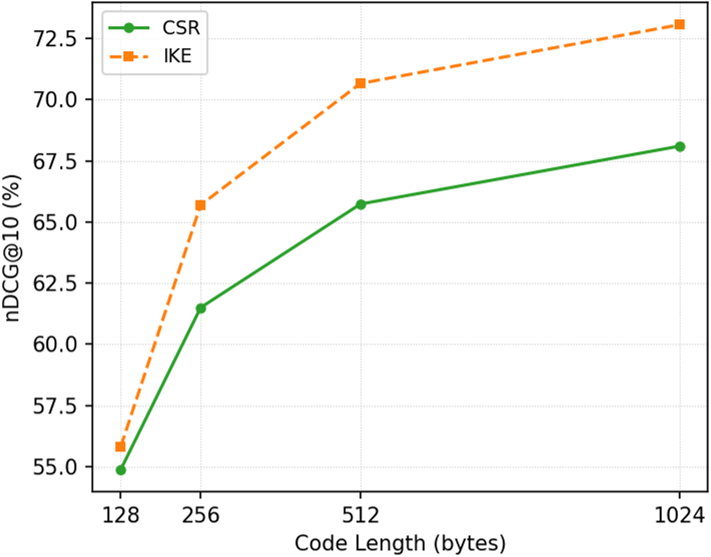}}\hfill
    \subfloat[][Search Time]{\label{fig:CSR_IKE_time}\includegraphics[width=.49\columnwidth]{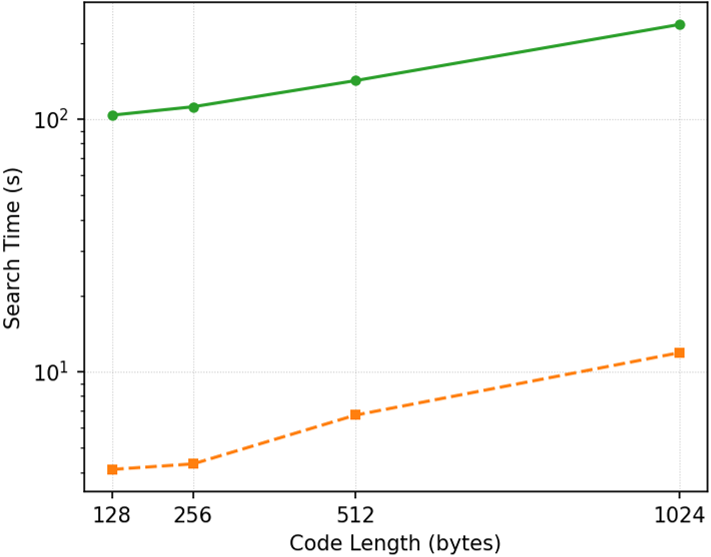}}

  \end{minipage}
  
  \caption{Comparison of Retrieval Accuracy and Search Time between CSR and {\modelname} in the exhaustive search setting on the HotpotQA (LLM2Vec) dataset.}
  \label{fig:CSR_IKE}
\end{figure}

\begin{figure*}[htbp]
  \centering
  
    \includegraphics[width=0.97\linewidth]{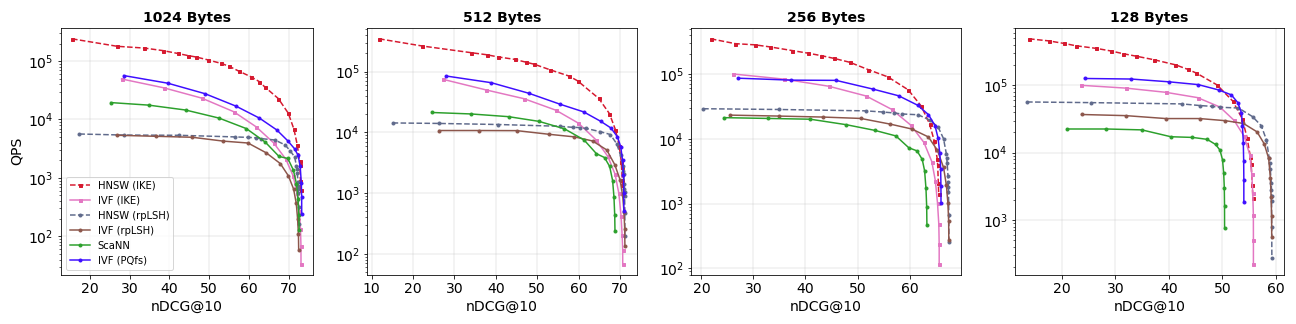}
    \label{fig:compression_hotpotqa_l2v}

  \caption{\textbf{Performance Comparison of Compression Methods on ANN Retrieval: QPS vs. nDCG@10 on the HotpotQA (LLM2Vec) dataset}. The curves are obtained by adjusting the search parameter of the index structures.}
  \label{fig:compression_ann_res}
\end{figure*}

\paragraph{CSR Parameter Settings} CSR uses a hyperparameter \(k\) to control the sparsity level of the vector \(z\), where each vector contains exactly \(k\) non-zero entries. We compare CSR with {\modelname} under fixed code length conditions, where the code length of a sparse vector \(z\) is defined as \(4k\) bytes. The CSR model is trained using the hyperparameters provided in the original paper for text retrieval tasks.  The search time includes both the mapping time for all queries and the similarity computation time, rather than solely the similarity calculation time as reported in the original CSR paper.

\paragraph{Comparative Results} Figure \ref{fig:CSR_IKE} presents a comparative evaluation in the exhaustive search setting on the HotpotQA (LLM2Vec) dataset between CSR and {\modelname} in terms of search time on a CPU and retrieval accuracy, measured by nDCG@10. Additional results on more datasets are provided in Appendix \ref{app:exp:learn_based}.  The results show that {\modelname} significantly outperforms CSR in terms of search efficiency across all datasets, while achieving comparable or even superior retrieval accuracy. Specifically, {\modelname} reduces the search time of CSR by one order of magnitude. On the HotpotQA (LLM2Vec) dataset, {\modelname} achieves a search time of 4–12 seconds, whereas CSR requires 104–237 seconds. This demonstrates that {\modelname} is a superior option to CSR in practical retrieval scenarios.



\paragraph{Discussion} We analyze the reasons behind the significantly higher retrieval latency of CSR compared to the {\modelname} method. The sparse vector mapping of CSR employs an MLP to transform the original \(d\)-dimensional embedding into a \(4d\)-dimensional space, after which only the top-\(k\) 
dimensions are retained while the remaining ones are set to zero. Specifically, the mapping process of CSR requires \(O(d^2 + d\log_2(k))\) operations, whereas the mapping in {\modelname} demands only \(O(t\lceil \log_2(\psi)\rceil)\), where \(\lceil \log_2(\psi)\rceil\) denotes the maximum iTree height.
Specifically, considering a scenario where $d=4096$ dimensional embedding is compressed into 128 bytes, {\modelname} utilizes hyperparameters $t=1024$ and $\psi=2$ while CSR requires $k=32$. The mapping complexity of CSR is dominated by the quadratic term $O(d^2)$, whereas that of {\modelname} scales linearly with $t$. Given that $d^2 \gg t$ in high-dimensional spaces, our method achieves a significant reduction in mapping overhead.
Additionally, for similarity computation, our method leverages highly optimized bitwise operations, which are faster than the sparse matrix multiplication relied upon by CSR.


\subsection{Comparison with compression methods}\label{sec:learningFree_with_index}
In this section, we compare the {\modelname} method with the learning-free vector compression techniques, under \textbf{same code-length} scenario in the ANN search setting. Specifically, we study different code length constraints of \{1024, 512, 256, 128\} bytes. 
For a fixed code length, we plot QPS against the retrieval effectiveness measure (nDCG@10) for varying ANN search parameters.
We compare our methods, IVF / HNSW (IKE), with several baselines: IVF / HNSW (rpLSH)~\cite{rpLSH}, IVF (PQfs)~\cite{PQfs}, and ScaNN~\cite{scann}. Detailed descriptions of these methods are provided in Appendix \ref{app:exp:learn_free:methods}. Note that the embedding length in {\modelname} is determined by parameters  $t$ and  $\psi$. Under a fixed code length, we find that smaller  $\psi$ and larger $t$ yield better retrieval performance. Thus, we set $\psi$=2 in this experiment.

\sloppy
Figure \ref{fig:compression_ann_res} illustrates QPS vs NDCG@10 curves for all methods on the HotpotQA (LLM2Vec) dataset, with the code lengths used indicated above each plot. Points closer to the upper right corner represent better overall performance. Our HNSW ({\modelname}) method achieves the best balance between efficiency and effectiveness across all datasets. In terms of search efficiency, it achieves from 2.9\(\times\) to 10\(\times\) higher throughput than other learning-free compression methods in an ANN setting, while maintaining comparable retrieval accuracy. For instance, on the HotpotQA (LLM2Vec) dataset with a fixed code length of 512 bytes, when achieving approximately nDCG@10 of 60\%, the IVF (PQfs) method attains around \(2.3\times10^4\) QPS, 
whereas the HNSW ({\modelname}) method achieves around \(6.7\times10^4\) QPS, representing a 2.9 times improvement in efficiency. Similarly, when the code length is extended to 1024 bytes to reach 70\% nDCG@10, HNSW ({\modelname}) attains $1.2 \times 10^4$ QPS, representing a 10\(\times\) efficiency improvement over IVF (rpLSH), which yields only $1.1 \times 10^3$ QPS. Additional results are provided in Appendix \ref{app:exp:learn_free:more_exp_res}. 


\begin{figure*}[htbp]
  \centering
  \begin{subfigure}{0.245\textwidth}
    \centering
    \includegraphics[width=\linewidth]{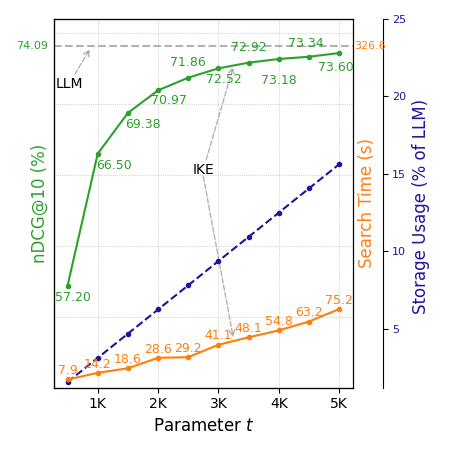}
    \caption{\centering Exhaustive Search\\ (HotpotQA, LLM2Vec)}
    \label{fig:para_t_hotpotqa_l2v_exhau}
  \end{subfigure}
  \begin{subfigure}{0.245\textwidth}
    \centering
    \includegraphics[width=\linewidth]{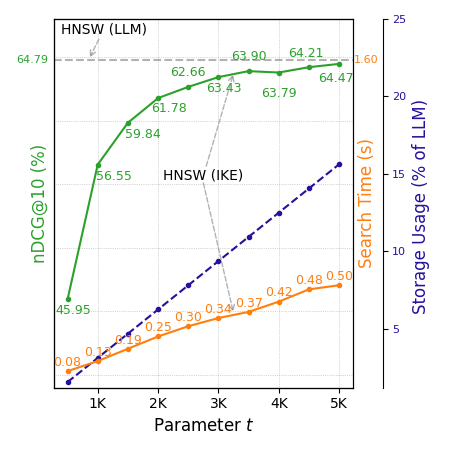}
    \caption{\centering HNSW\\ (HotpotQA, LLM2Vec)}
    \label{fig:para_t_hotpotqa_l2v_hnsw}
  \end{subfigure}
  \begin{subfigure}{0.245\textwidth}
    \centering
    \includegraphics[width=\linewidth]{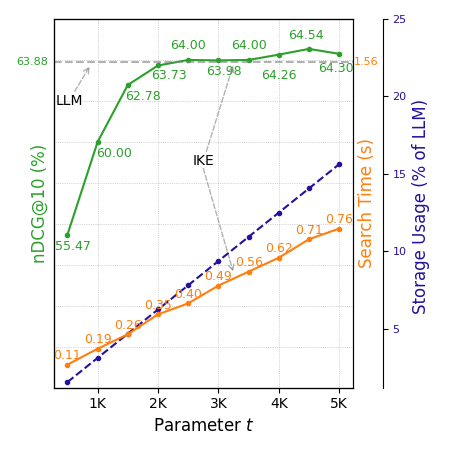}
    \caption{\centering Exhaustive Search\\ (Istella22, LLM2Vec)}
    \label{fig:para_t_istella22_exhau}
  \end{subfigure}
  \begin{subfigure}{0.245\textwidth}
    \centering
    \includegraphics[width=\linewidth]{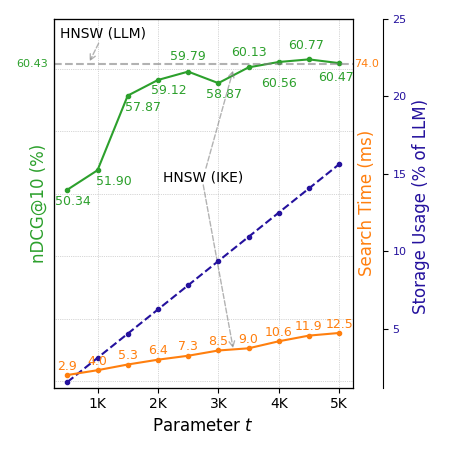}
    \caption{\centering HNSW\\ (Istella22, LLM2Vec)}
    \label{fig:para_t_istella22_hnsw}
  \end{subfigure}
  \caption{\textbf{Effect of parameter \(t\) (the number of iTrees) on Search Time and Retrieval Accuracy}. The gray dashed lines denote the corresponding retrieval performance in the LLM embedding space. 
}
  \label{fig:para_t}
\end{figure*}


\subsection{Extra analyses} \label{sec:extra_ana}

\subsubsection{Impact of hyperparameter \(t\)}
The hyperparameter \(t\) denotes the number of iTrees in the iForest. Increasing \(t\) produces more fine‐grained representations, which can enhance retrieval accuracy; conversely, reducing \(t\) trades off accuracy for lower query latency and reduced memory usage. To assess this trade‐off, we evaluate the impact of varying \(t\) on the HotpotQA (LLM2Vec) and Istella22 (LLM2Vec) datasets under both exhaustive search and HNSW configurations. For the HNSW setup, we use a search parameter $efSearch=50$. We draw \(t\) from 0.5K to 5K (where 1K = 1024), and at each value, we record the search time, retrieval accuracy (nDCG@10), and storage usage relative to LLM embeddings.

The results in Figure \ref{fig:para_t} show that search time increases approximately linearly with the number of iTrees \(t\), whereas retrieval accuracy follows a logarithmic growth pattern. Specifically, while each increment in \(t\) below 3K yields noticeable gains in retrieval accuracy, the marginal improvement diminishes beyond this threshold. For instance, as illustrated in Figure \ref{fig:para_t_hotpotqa_l2v_hnsw}, reducing \(t\) from 4K to 3K decreases query latency from 0.42s to 0.34s, corresponding to a 19.05\% improvement in search efficiency and a 25\% reduction in memory footprint, at the expense of only a 0.36\% drop in nDCG@10 (from 63.79\% to 63.43\%). Thus, tuning \(t\) enables substantial gains in search efficiency and memory savings with a minimal impact on retrieval accuracy. \textit{In other words, \modelname{} exhibits a desirable property similar to that of MRL, enabling flexible adaptation of code length to balance effectiveness and efficiency, while requiring no costly training.}

\subsubsection{Impact of hyperparameter \(\psi\)}
The hyperparameter $\psi$ denotes the number of isolating partitions within an iTree. To investigate its influence, we conduct a sensitivity analysis on the validation set of the Istella22 (LLM2Vec) dataset using exhaustive search. Specifically, we fix the number of iTrees at $t = 4096$ and vary $\psi$ across the set $\{2, 4, 8, 16\}$. We report the resulting computational overhead and retrieval precision, as summarized in Table \ref{tab:para_psi}. 

The experimental results indicate that, unlike the monotonic accuracy gains observed with parameter $t$ in Figure \ref{fig:para_t}, the impact of $\psi$ on retrieval precision is relatively marginal and lacks a consistent monotonic trend. Notably, however, even with small values of $\psi$, \modelname{} maintains remarkably high retrieval accuracy while reducing both mapping overhead and query latency. These findings provide critical engineering guidance: in latency-critical scenarios, selecting a small $\psi$ enables practitioners to maximize computational efficiency with negligible compromises in retrieval performance.

\begin{table}[!ht]
    \centering
    \caption{Sensitivity analysis for parameter \(\psi\)}
    \label{tab:para_psi}
        \resizebox{\columnwidth}{!}{
    \begin{tabular}{cccccc}
        \toprule
         \multirow{2}{*}{\(t\)} & \multirow{2}{*}{\(\mathbf{\psi}\)} & \multicolumn{2}{c}{\textbf{Time (s)}} & \multirow{2}{*}{\textbf{MRR@10}} & \multirow{2}{*}{\textbf{nDCG@10}} \\
        \cmidrule(lr){3-4}
         & & \textbf{Other} & \textbf{Search} & & \\
        \midrule
         4096 & 2 & 1.18 & 0.16 & 91.06 & 75.71 \\
         4096 & 4 & 1.69 & 0.25 & 90.99 & 75.83 \\
         4096 & 8 & 2.49 & 0.35 & 90.94 & 76.02 \\
         4096 & 16 & 3.70 & 0.36 & 91.43 & 76.35 \\
        \bottomrule
    \end{tabular}
    }
\end{table}


        

        

\begin{table}[!ht]
    \centering
    \caption{Results of different implementations of the {\modelname} in the exhaustive search setting.}
    \label{tab:vd_impl}
        \resizebox{\columnwidth}{!}{
    \begin{tabular}{ccc}
        \toprule
         \textbf{Method} & \textbf{MRR@10 (\%)} & \textbf{nDCG@10 (\%)} \\
        
        \midrule
        \multicolumn{3}{c}{FiQA2018 (LLM2Vec)}\\
        \cmidrule(lr){1-3}
         {\modelname} ($\psi$=6) & 60.13 & 52.58 \\
  {\modelnamevd} ($\psi$=3, $m$=7) & 60.23 & 52.89 \\
 {\modelnamevd} ($\psi$=3, $m$=1) & 59.76 & 52.61 \\
 \cdashline{1-3}
VDeH ($\psi$=3) & 32.07 & 25.84 \\

        \cmidrule(lr){1-3}
        \multicolumn{3}{c}{FiQA2018 (Qwen3)}\\
         \cmidrule(lr){1-3}
 {\modelname} ($\psi$=15) & 71.14 & 63.25  \\
  {\modelnamevd} ($\psi$=4, $m$=5)  & 71.99 & 64.03 \\
  {\modelnamevd} ($\psi$=4, $m$=1)  & 69.87 & 62.16 \\ \cdashline{1-3}
 VDeH ($\psi$=4)  & 58.91 & 49.78 \\
        
        \bottomrule
    \end{tabular}
    }
\end{table}

\sloppy
\subsubsection{Variations of IK implementations}

In this section, we evaluate the performance of VDeH, \modelnamevd{} (\(m\leq d\)) and \modelname{}. Here, we report retrieval performance under three configurations of \(m\): \(m=1\), \(m=d\) (equivalent to VDeH), and an optimally tuned \(m\), obtained through hyperparameter search. The value of \(t\) is fixed at 4096, while \(\psi\) is tuned for best retrieval performance.

The results, presented in Table \ref{tab:vd_impl}, reveal a substantial performance gap between {\modelname} and {\modelnamevd} with \(m=d\). For example, in terms of nDCG@10, {\modelnamevd} (\(m=d\)) exhibits performance decreases of 50.86\% and 21.30\% on the two datasets, respectively. However, when \(m\) is reduced to smaller values such as \(m=1\) or the optimally tuned \(m\) (which also tends to be small), the performance gap almost disappears. \modelnamevd{} (\(m=1\)) is slightly worse than \modelname{}, which is aligned with the analysis in Appendix \ref{app:diverse_partitioning}, yet \modelnamevd{} with optimally tuned $m$ can surpass \modelname{} in MMR@10 and nDCG@10. Nevertheless, we recommend {\modelname} over the optimally tuned {\modelnamevd} for two primary reasons. First, {\modelnamevd} requires tuning more hyperparameters (i.e., the extra parameter $m$ beside $t, \psi$). Second, the mapping process of {\modelnamevd} is computationally more expensive compared to \modelname{}.


\section{Conclusion}


This paper presents \modelname, a learning-free method based on the Isolation Kernel that transforms LLM embeddings into binary representations (in IKE space). \modelname{} ensembles diverse partitions to reduce retrieval loss as the ensemble grows. It is lightweight, binary-encoded, and adaptable in code length, offering low memory use and fast mapping. Theoretical analysis highlights the importance of \emph{high diversity} in \modelname{} and its variant \modelnamevd{}.

Experiments on multiple text retrieval datasets show that {\modelname} cuts memory and search time by up to 16.7× while maintaining or even surpassing LLM-level accuracy. Compared to the learning-based CSR, {\modelname} achieves order-of-magnitude speedups without retraining. When combined with HNSW, it outperforms other learning-free compression methods (rpLSH, ScaNN, PQfs) in the ANN setting, delivering from 2.9\(\times\) to 10\(\times\) higher throughput with comparable accuracy.

\section*{Limitations}
While our method demonstrates strong performance on the text retrieval task, its effectiveness on cross-modal retrieval tasks (e.g., Image‑to‑Text or Text‑to‑Image) remains limited. This may be attributed to the modality gap~\cite{liang2022mind} between textual and visual representations, which poses a challenge for cross‑modal alignment. This aspect has not been investigated in this paper and is left for future research.

\section*{Acknowledgements}
Kai Ming Ting is supported by the National Natural Science Foundation of China (Grant No. 92470116 \& W2531050).
This project is also supported by the National Natural Science Foundation of China (Grant No. W2532049) and the Fundamental and Interdisciplinary Disciplines Breakthrough Plan of the Ministry of Education of China (No. JYB2025XDXM118).

\bibliography{sample}
\clearpage
\appendix

\section{Related Work}
\subsection{LLM Embedded Representation}
Text representation plays a crucial role in tasks such as semantic matching, clustering, and information retrieval \cite{nie2024text-emb-llm}. In recent years, pretrained language models with millions or even billions of parameters have been proposed for a wide range of NLP tasks, including text representation. Early approaches relied on encoder-based models such as BERT \cite{bert} and RoBERTa \cite{roberta}. More recently, there has been a paradigm shift toward leveraging large language models (LLMs) for text embeddings, with notable examples including LLM2Vec \cite{llm2vec}, Llama2Vec \cite{llama2vec}, and Qwen Embedding \cite{qwen-gte,qwen3embedding}. These LLM-based methods now represent the state of the art in text representation.

However, language model-based embeddings typically have high dimension: for example, the hidden states of T5-11B \cite{t5-11b} have 1024 dimensions, while mainstream 7B decoder-only LLMs use 4096 dimensions. Such high-dimensional embeddings substantially increase storage and inference costs. To address this, \citet{kusupati2022matryoshka} introduced a flexible technique that enables fine-grained control over embedding size, balancing latency and accuracy. This approach has since been extended to various applications \cite{openai2024embedding,nussbaum2025nomic,yu2024arctic}. Recently, \citet{wenbeyond} proposed CSR, which combines sparse coding with contrastive learning to produce efficient sparse embeddings with minimal performance loss. 

Similar to \cite{kusupati2022matryoshka,wenbeyond}, our study also targets representation efficiency and is complementary to the aforementioned LLM-based representation. However, while these studies rely on supervised learning to reduce performance degradation relative to LLM embeddings, our method does not require learning and instead leverages IK for transformation. 

\subsection{Efficient Retrieval}
For efficient retrieval in high-dimensional spaces, techniques such as dimensionality reduction \cite{van2009dimensionality}, product quantization \cite{scann, PQ, PQfs}, and hashing \cite{rpLSH, xu2025VDeH} are commonly applied, though often at the expense of accuracy. Our study falls within the hashing paradigm, yet demonstrates strong potential to preserve—and in some cases even improve—retrieval accuracy.

Hashing techniques can be broadly categorized into training-based approaches (Learning to Hash, L2H) \cite{rcLSH,SH,CBE,ITQ,BP,SP,DSH,SELVE,SpH} and learning-free approaches \cite{rpLSH,cai2019revisit,xu2025VDeH}. L2H aims to map high-dimensional vectors into a binary coding space for efficient processing and retrieval. In general, learning is to ensure three key criteria \cite{xu2025VDeH} for an effective hashing: full space coverage, entropy maximization, and bit independence. 

Recently, \citet{xu2025VDeH} showed that Voronoi Diagrams (VD) offer a natural alternative to L2H, as they inherently satisfy these three criteria. Building on this insight, they proposed VDeH, a learning-free method that leverages a VD–based implementation of the Isolation Kernel for effective hashing. Though VDeH has been shown to outperform many L2H methods for non-LLM-related databases, our early examination has revealed that VDeH performs poorly in LLM-embedded databases (see Figure \ref{fig:para_m}). To address this limitation, we propose a generalized version of VDeH that is amenable to LLM-embedded databases in this paper. It is worth noting that prior work \cite{relevanceFeatureMapping} also explored the application of Isolation Forest in the field of retrieval. Their method employs the path length of data points on each Isolation Tree as the mapped feature representation, along with a weighted‑sum similarity calculation, rather than adopting a hashing paradigm. This, however, results in high computational overhead and difficulty in low‑bit compression.

\section{Extension of section \ref{sec:preliminaries}}\label{app:sec:preliminaries}
\subsection{Detailed Description of Isolation Kernel}\label{app:IK_detail}


In practice, Isolation Kernel $K_\psi$ is usually estimated from a finite number of $t$ partitions $H_i \enspace (i=1,\dots,t)$ and each $H_i$ is derived from a data subset $\mathcal{D}_i\subset D$:
\begin{equation}
\begin{aligned}
    K_\psi(\mathbf{x},\mathbf{y} \mid D) 
    &\simeq \frac{1}{t} \,\sum_{i=1}^t \mathbb{I}(\mathbf{x},\mathbf{y} \in \theta \mid \theta \in H_i) \\
    &= \frac{1}{t} \, \sum_{i=1}^t \sum_{\theta_j \in H_i} \mathbb{I}(\mathbf{x} \in \theta_j)\mathbb{I}(\mathbf{y} \in \theta_j).
\end{aligned}
\end{equation}
Given \(H_i\), there exist \(\psi\) isolating partitions \(\theta_j \in H_i \enspace\) \((j\)=\(1,\dots,\psi)\). For any data point \(\mathbf{x}\), we define
\begin{displaymath}
\Phi_{ij}(\mathbf{x}\mid D)
=
\mathbb{I}\bigl(\mathbf{x} \in \theta_j \,\bigm|\; \theta_j \in H_i\bigr).
\end{displaymath}
Since \(\mathbf{x}\) can belong to exactly one of these \(\psi\) isolating partitions, the vector
\begin{align*}
&\Phi_{i}(\mathbf{x}\mid D)
\\&=
\bigl[
\Phi_{i1}(\mathbf{x}\mid D),\;\Phi_{i2}(\mathbf{x}\mid D),\;\dots,\;\Phi_{i\psi}(\mathbf{x}\mid D)\bigr]
\end{align*}
is a \(\psi\)-dimensional one-hot encoding. Given $t$ partitions, $\Phi(\mathbf{x} \mid D)$ is the concatenation of $\Phi_1(\mathbf{x} \mid D),\dots, \Phi_t(\mathbf{x} \mid D)$.



In practice, each data point \(\mathbf{x}\) is mapped into a binary feature space of dimension \(t \times \psi\) via the mapping
\(
\Phi: \mathbf{x} \;\longmapsto\; \{0,1\}^{\,t \times \psi}.
\)
Under this representation, the Isolation Kernel \(K_\psi\) can be approximated by the normalized inner product:
\begin{equation}
\begin{aligned}
    K_\psi\bigl(\mathbf{x},\mathbf{y} \mid D\bigr)
    \simeq
    \frac{1}{t}\,\bigl\langle \Phi\bigl(\mathbf{x} \mid D\bigr),\,\Phi\bigl(\mathbf{y} \mid D\bigr)\bigr\rangle.
\end{aligned}
\end{equation}
For notational brevity, we henceforth omit the conditioning on \(D\) whenever the meaning is clear. 

Since the mapped data point \(\Phi(\mathbf{x})\) 
is highly sparse, with only \(t\) out of the \(t\times\psi\) entries equal to 1, we store it in an index‐based sparse representation 
\(\Phi^{\mathrm{idx}}(\mathbf{x})\in\mathbb{N}^t\). 
Specifically, the \(i\)-th component of \(\Phi^{\mathrm{idx}}(\mathbf{x})\) indicates the position of the sole ``1'' in the one‐hot vector \(\Phi_i(\mathbf{x})\) for all \(i\in [1,t]\). Formally, we define
\begin{displaymath}
\Phi^{\mathrm{idx}}(\mathbf{x})
=
\bigl[
\Phi_{1}^{\mathrm{idx}}(\mathbf{x}),\;\Phi_{2}^{\mathrm{idx}}(\mathbf{x}),\;\dots,\;\Phi_{t}^{\mathrm{idx}}(\mathbf{x})\bigr]
\end{displaymath}
\begin{displaymath}
\Phi^{\mathrm{idx}}_i(\mathbf{x})
\;=\;
\arg\max_{\,1\le j\le \psi\,} \Phi_{ij}(\mathbf{x}).
\end{displaymath}
By storing each element of \(\Phi^{\mathrm{idx}}(\mathbf{x})\) using \(\lceil \log_2(\psi) \rceil\) bits, the code length of each mapped data point becomes \(t\lceil \log_2(\psi) \rceil\) bits, significantly reducing the memory requirements. Thus, we can approximate the Isolation Kernel \(K_\psi\) as:
\begin{equation} \label{eq:Kpsi_idx_appendix}
\begin{aligned}
    K_\psi\bigl(\mathbf{x},\mathbf{y} \mid D\bigr)
    \simeq
    \frac{1}{t}\,\sum_{i=1}^t\mathbb{I}(\Phi^{\mathrm{idx}}_i(\mathbf{x})=\Phi^{\mathrm{idx}}_i(\mathbf{y})).
\end{aligned}
\end{equation}

\section{Theoretical Guarantees for IKE}
\label{app:ike-theory}

This section formally defines the Voronoi Diagram-based Isolation Kernel Embedding ({\modelnamevd}) operating on a randomly selected subspace (see Appendix~\ref{app:def_of_ike_vd}) and establishes that it preserves the three key criteria demonstrated for VDeH~\cite{xu2025VDeH}: \textbf{full space coverage}, \textbf{entropy maximization}, and \textbf{bit independence}. The proofs follow directly from the theoretical framework of~\cite{xu2025VDeH}; since the core mechanisms of random sampling and partitioning remain unchanged, we show that adding a random subspace projection does not affect the original guarantees. For brevity, we refer to the corresponding proofs and streamline the arguments where the reasoning is identical (see Appendix~\ref{app:properties_of_ike_vd}).

Furthermore, under suitable distributional assumptions, we prove that the iForest-based IKE also satisfies these three criteria (see Appendix~\ref{app:properties_of_ike}). Finally, by viewing both VDeH and \modelname{} as ensemble methods, we introduce \textbf{diverse partitioning} as a fourth criterion. We demonstrate that, compared to VDeH, \modelname{} produces more robust hash representations due to its ability to generate a more diverse set of base partitioners (see Appendix~\ref{app:diverse_partitioning}).


\subsection{Formal Definition of {\modelnamevd}} \label{app:def_of_ike_vd}

Let $\mathcal{X}=[x_1,\dots,x_N]\in\mathbb{R}^{d\times N}$ be the input dataset. An {\modelnamevd} partition is constructed via a three-step process: (1) We randomly select $m$ dimensions from the original $d$-dimensional embedding, which is equivalent to applying a projection function $P_J:\mathbb{R}^{d} \to \mathbb{R}^{m}$ that maps points to a random $m$-dimensional subspace; (2) Randomly sampling $\psi$ anchor points $\mathcal{D}=\{s_1,\dots,s_{\psi}\}\subset\mathcal{X}$; (3) Constructing a Voronoi Diagram $H_J$ in the $\mathbb{R}^{m}$ subspace using the projected anchor points $P_J(\mathcal{D})$. For any point $x\in\mathbb{R}^{d}$, its hash value is determined by which of the $\psi$ Voronoi cells its projection $P_{J}$ falls into. An {\modelnamevd} model consists of $t$ such partitions, each constructed independently.

\subsection{Proof of Criteria of {\modelnamevd}}\label{app:properties_of_ike_vd}

\begin{prop}[Full Space Coverage]
    The hashing scheme of {\modelnamevd} covers the entire input space, ensuring that every point can be assigned a binary code.
\end{prop}

\begin{proof}
    The {\modelnamevd} scheme inherently provides full space coverage. For any point $x\in \mathbb{R}^{d}$, the projection $P_J$ maps it to a unique point $P_J(x)$ in the subspace $\mathbb{R}^{m}$. By definition, a Voronoi Diagram forms a complete partition of the space in which it is constructed. Therefore, $P_J(x)$ is guaranteed to fall into exactly one Voronoi cell. This ensures that every point $x$ in the original $d$-dimensional space is assigned a hash code.
\end{proof}

\begin{prop}[Entropy Maximization]
    For a given dataset, every hash function in {\modelnamevd} covers approximately the same number of data points.
\end{prop}

\begin{proof}
    Based on Lemma 1 of the existing work~\cite{xu2025VDeH}, for a partition constructed from $\psi$ anchor points drawn \textit{i.i.d.} from the data distribution, the probability of an arbitrary point falling into any given Voronoi cell is uniformly $1/\psi$. 
    
    This criterion remains valid for {\modelnamevd}. The anchor points are still sampled \textit{i.i.d.} from the data distribution. When these points and the data are projected onto the same random subspace $\mathbb{R}^{m}$, their underlying distributional relationship is preserved. The entropy maximization criterion of VDeH is based on the random sampling process and is agnostic to the space's dimensionality. Thus, its conclusion holds true within the subspace $\mathbb{R}^{m}$:
    $$\text{Pr}(P_J(x) \in V_i) = \frac{1}{\psi},$$
    for every cell \(i=1, \dots, \psi\). This uniform probability ensures entropy maximization.
\end{proof}

\begin{prop}[Bit Independence]
    All bits in the final concatenated binary code generated by {\modelnamevd} are mutually independent.
\end{prop}

\begin{proof}
    The bit independence criterion in VDeH relies on two aspects, both of which are maintained in {\modelnamevd}. We assume the standard encoded hashing mechanism, which maps a cell index from a partition of size $\psi=2^w$ to a $w$-bit binary code.

    First, {\modelnamevd} holds the independence between partitions. As with VDeH, each of the $t$ partitions in {\modelnamevd} is constructed independently. This independence is further reinforced in our case by the fact that each partition is constructed not only with a new set of random anchor points but also in a new, independently chosen random subspace. This ensures the statistical independence of the $w$-bit codes generated from different partitions.
    Second, {\modelnamevd} holds the independence within a partition. The proof for this case, detailed in Theorem 1 of the existing work~\cite{xu2025VDeH}, is conditional on the entropy maximization criterion holding. Specifically, it requires the probability of a point falling into any of the $\psi$ cells to be $1/\psi$, which is also satisfied in {\modelnamevd}. Since the prerequisite holds and the same encoding logic is applied, the conclusion of Theorem 1 in VDeH directly applies: any two bits generated from the same encoded partition are mutually independent.
\end{proof}

\subsection{Proof of Criteria of IKE (iForest)}\label{app:properties_of_ike}

For iForest-based IKE, as with {\modelnamevd}, the full space coverage criterion is naturally satisfied, while the bit independence criterion depends on the entropy maximization.

First, an iTree partitions space using axis-aligned hyperplanes. The set of all leaf nodes forms a collection of disjoint hyper-rectangles whose union completely tiles the $\mathbb{R}^d$ space. Any point $x \in \mathbb{R}^d$ can be deterministically routed through the tree's splits to a unique leaf node. Thus, full space coverage is satisfied, irrespective of the data distribution.

\begin{thm}
    Under the assumption of a uniform data distribution, the hashing scheme of IKE satisfies the entropy maximization and the bit independence criteria.
\end{thm}

\begin{proof}
    Let an iTree be constructed to have $\psi$ leaf nodes $\{V_1, \dots, V_{\psi}\}$. We need to show that for an arbitrary point $x$, the probability of it landing in any leaf node $V_i$ is uniform, i.e., $\text{Pr}(x \in V_i) = 1/\psi$.
    Consider a single split at an arbitrary node in the tree. The split value $p$ is chosen uniformly at random from the range $[m_q,M_q]$, where $m_q$ and $M_q$ are the minimum and maximum values of the data points within that node along the chosen dimension $q$. Under the uniform data distribution, the probability mass within any sub-region is proportional to its volume. When a split point $p$ is chosen uniformly within the range $[m_q,M_q]$, the probability that a point $x$ within this region falls on the left side of the split ($x_q<p$) is given by $\frac{p - m_q}{M_q - m_q}$. The expected probability of going left, over the random choice of $p$, is:
    \begin{align*}
        E_p\left[\text{Pr}(x \text{ goes left})\right] &= E_p\left[\frac{p - m_q}{M_q - m_q}\right] \\&= \frac{E_p[p] - m_q}{M_q - m_q}.
    \end{align*}
    Since $p$ is uniform in $[m_q,M_q]$, its expectation is $E_p[p] = \frac{m_q + M_q}{2}$. Substituting this in, we get:
    \begin{align*}
        E_p\left[\text{Pr}(x \text{ goes left})\right] &= \frac{\frac{m_q + M_q}{2} - m_q}{M_q - m_q} \\&= \frac{\frac{M_q - m_q}{2}}{M_q - m_q} = \frac{1}{2}.
    \end{align*}
    Thus, at each node, a data point has an equal probability of $1/2$ of being routed to either child node. For a point to reach a specific leaf node, it must follow a unique path of splits from the root. The probability of any specific path is therefore uniform across all leaves, leading to $\text{Pr}(x \in V_i) = 1/\psi$, which establishes the entropy maximization criterion.
    
    The bit independence criterion follows from the establishment of entropy maximization. Assuming the standard encoded hashing mechanism where a leaf index from a partition of size $\psi=2^w$ is mapped to a $w$-bit binary code. The proof considers two cases: (i) The independence between partitions is guaranteed because each of the $t$ iTrees in the ensemble is constructed independently. This ensures the statistical independence of the $w$-bit codes generated from different trees. (ii) The independence within a single partition. Let $\alpha$ and $\beta$ be two distinct bits in the $w$-bit vector generated from one iTree. For any values $v_\alpha, v_\beta \in \{0,1\}$, the joint probability is the sum of probabilities of landing in all leaf nodes $V_i$ whose code satisfies the condition: 
    \begin{align*} 
    &\text{Pr}(\alpha=v_\alpha, \beta=v_\beta) \\&=\sum_{i=1}^{\psi} \text{Pr}(x \in V_i) \cdot \mathbb{I}(\alpha=v_\alpha, \beta=v_\beta \text{ for cell } i) \\ &= \frac{1}{\psi} \sum_{i=1}^{\psi} \mathbb{I}(\alpha=v_\alpha, \beta=v_\beta).
    \end{align*} 
    Since the $\psi=2^w$ binary codes exhaust all possible $w$-bit vectors, exactly $\psi/4$ of them will satisfy the condition that two specific bits have the values $v_\alpha$ and $v_\beta$. Therefore: 
    $$ \text{Pr}(\alpha=v_\alpha, \beta=v_\beta) = \frac{1}{\psi} \left(\frac{\psi}{4}\right) = \frac{1}{4}.$$
    Similarly, the marginal probabilities are: 
    $$ \text{Pr}(\alpha=v_\alpha) = \frac{1}{\psi} \sum_{i=1}^{\psi} \mathbb{I}(\alpha=v_\alpha) = \frac{1}{\psi} \left(\frac{\psi}{2}\right) = \frac{1}{2},$$ 
    $$ \text{Pr}(\beta=v_\beta) = \frac{1}{\psi} \sum_{i=1}^{\psi} \mathbb{I}(\beta=v_\beta) = \frac{1}{\psi} \left(\frac{\psi}{2}\right) = \frac{1}{2}.$$ 
    Since $\text{Pr}(\alpha=v_\alpha, \beta=v_\beta) = \text{Pr}(\alpha=v_\alpha) \cdot \text{Pr}(\beta=v_\beta)$, the bits are mutually independent. This completes the proof.
\end{proof}

Therefore, under the uniform distribution, IKE satisfies three key criteria, which is to implement effective binary codes.

In fact, LLM embedding is non-uniform, which makes IKE appear to perform poorly due to not strictly satisfying the two criteria. However, IKE's strength lies in its use of a large ensemble of $t$ iTrees. While a single tree may produce biased partitions and correlated bits, the final similarity score is an aggregation over thousands of independent random processes. The biases and errors of individual trees are averaged out, yielding a robust model where the strong independence between trees compensates for the lack of strict independence within them. In our experience, IKE is even better than {\modelnamevd} when $m=1$, which satisfies all three criteria.

\subsection{Diverse Partitioning}\label{app:diverse_partitioning}
It is noteworthy that VDeH satisfies the three criteria discussed above, yet it performs poorly on LLM embeddings. This indicates that these criteria alone are insufficient for effective retrieval. In the following, we demonstrate that diverse partitioning accounts for the superior retrieval performance of \modelname{} compared to VDeH.

\paragraph{\textbf{Diverse partitioning as the fourth criterion.}} As VDeH and IKE can be expressed as ensemble methods of a similarity kernel function $K(\mathbf{x}, \mathbf{y})$~\cite{ting2018isolation,ting2024nearest}, the superiority of IKE stems from a key principle in ensemble learning: diverse, simple base learners often outperform ensembles of complex, correlated ones. 

\begin{definition}[Base Partitioner] 
A base partitioner $h(\mathbf{x}; \Theta)$ maps $\mathbf{x} \in \mathbb{R}^d$ to a discrete partition, with $\Theta$ as a random variable capturing all randomness:
\begin{itemize} 
\item For VDeH, $\Theta_{\text{VDeH}}$ consists of $\psi$ randomly sampled anchor points from the corpus. 
\item For IKE, $\Theta_{\text{IKE}}$ includes random dimension selections, and value splits during iTree construction. \end{itemize} 
\label{def:base_partitioner} 
\end{definition} 

\begin{definition}[Isolation Kernel] 
Given a base partitioner $h(\cdot; \Theta)$ (Definition~\ref{def:base_partitioner}), the Isolation Kernel $K_\psi(\mathbf{x}, \mathbf{y})$ is the expectation that $\mathbf{x}$ and $\mathbf{y}$ are mapped to the same partition:

\begin{equation} 
K_\psi(\mathbf{x}, \mathbf{y}) = \mathbb{E}_{\Theta}[\mathbb{I}(h(\mathbf{x}; \Theta) = h(\mathbf{y}; \Theta))],
\end{equation} 
\noindent This is consistent with the definition in Section~\ref{sec:preliminaries}, re-expressed using the base partitioner notation.
\label{def:ideal_kernel} 
\end{definition} 

Intuitively, if two points $\mathbf{x}$ and $\mathbf{y}$ are close in the data space, they are more likely to fall into the same isolating partition, yielding a larger kernel value $K_\psi(\mathbf{x}, \mathbf{y})$. Conversely, distant points are less likely to co-occur in a partition, yielding a smaller value.

In practice, both IKE and VDeH estimate $K_\psi(\mathbf{x}, \mathbf{y})$ using a finite ensemble of $t$ base partitioners, $\{h_1, \dots, h_t\}$. Each partitioner $h_i$ is an instance $h(\cdot; \Theta_i)$, generated from an i.i.d.\ realization $\Theta_i$ of the random parameter $\Theta$: \begin{equation} K(\mathbf{x}, \mathbf{y}; t) = \frac{1}{t} \sum_{i=1}^{t} \mathbb{I}(h(\mathbf{x}; \Theta_i) = h(\mathbf{y}; \Theta_i)). \end{equation}
The quality of the estimator $K(\mathbf{x}, \mathbf{y}; t)$ is measured as follows:
\begin{equation}
\begin{split}
    \text{MSE}(K(\mathbf{x}, \mathbf{y}; t)) &= \mathbb{E}\left[(K(\mathbf{x}, \mathbf{y}; t) - K_\psi(\mathbf{x}, \mathbf{y}))^2\right] \\
    &= \text{Bias}^2 + \text{Variance}.
\end{split}
\label{eq:mse_decomp}
\end{equation}

Since each $Z_i = \mathbb{I}(h(\mathbf{x}; \Theta_i) = h(\mathbf{y}; \Theta_i))$ is i.i.d.\ with $\mathbb{E}[Z_i] = K_\psi(\mathbf{x}, \mathbf{y})$, the sample mean $K(\mathbf{x}, \mathbf{y}; t) = \frac{1}{t}\sum_{i}Z_i$ is an unbiased estimator of $K_\psi(\mathbf{x}, \mathbf{y})$. Thus, $\text{MSE}(K(\mathbf{x}, \mathbf{y}; t))$ is determined entirely by the variance term. Let $p = K_\psi(\mathbf{x}, \mathbf{y})$ and $\sigma^2 = p(1-p)$. The variance of the estimator is
\begin{equation}
\begin{split}
    &\text{Var}(K(\mathbf{x}, \mathbf{y}; t)) \\
    &= \frac{1}{t^2} \left( \sum_{i=1}^{t} \text{Var}(Z_i) + \sum_{i \neq j} \text{Cov}(Z_i, Z_j) \right).
\end{split}
\label{eq:bernoulli_var}
\end{equation}

Letting $\rho = \text{Cov}(Z_i, Z_j)/\sigma^2$ for $i \neq j$ gives the following standard identity for the variance of the mean of identically distributed random variables~\cite{breiman2001random}:
\begin{equation}
    \text{Var}(K(\mathbf{x}, \mathbf{y}; t)) = \sigma^2 \left( \frac{1-\rho}{t} + \rho \right).
\label{eq:breiman_variance}
\end{equation}

\begin{col}[Asymptotic Variance Limit]
As the ensemble size $t \to \infty$, the variance of $K(\mathbf{x}, \mathbf{y}; t)$ converges to:
\begin{equation}
    \lim_{t\to\infty} \text{Var}(K(\mathbf{x}, \mathbf{y}; t)) = \rho\sigma^2
\label{eq:asymptotic_variance}
\end{equation}
\end{col}

Thus, the quality of the estimator $K(\mathbf{x}, \mathbf{y}; t)$ is ultimately constrained by the correlation $\rho$.

\begin{thm}
Assuming that the dimensions in the input space are independent, the correlations of the VDeH, {\modelnamevd}, and IKE partitioning schemes satisfy the following inequality for dimension $d>1$:
\begin{equation}
\rho_{\text{VDeH}} > \rho_{\text{IKE}_{\text{VD}}(m=1)} > \rho_{\text{IKE}} \ge 0.
\end{equation}
\label{thm:rho}
\end{thm}
\begin{proof}
    For VDeH, the randomness is driven solely by the data sample $D_\psi$. Therefore, its correlation is entirely determined by the data, yielding $\rho_{\text{VDeH}} = \rho_{\text{data}}$. For \modelnamevd{}$(m=1)$, randomness comes from $D_\psi$ and a random dimension choice $q \in \{1,\dots,d\}$. We analyze its covariance $Cov_{VD,1}(Z_i, Z_j)$ using the law of total covariance, conditioning on the dimension choices $q_i, q_j$:
    \begin{equation}
\begin{split}
    &\mathbb{E}_{q_i,q_j}[Cov(Z_i, Z_j | q_i, q_j)] \\ 
    &\quad + Cov_{q_i,q_j}(\mathbb{E}[Z_i|q_i], \mathbb{E}[Z_j|q_j]).
\end{split}
\end{equation}

    By dimension independence, $\mathbb{E}[Z|q]$ is constant, making the second term zero.  The covariance is thus equal to the first term:
    \begin{itemize}
    \item {Case A ($q_i = q_j$):} Occurs with probability $1/d$, giving \\ $Cov(Z_i, Z_j | q_i=q_j) = \sigma^2 \rho_{\text{data}}$.
    \item {Case B ($q_i \neq q_j$):} Occurs with probability $1 - 1/d$, yielding zero conditional covariance due to dimension independence.
    \end{itemize}
    
    Combining these cases yields:
    \begin{equation}
    \rho_{\text{IKE}_{\text{VD}}(m=1)} = \frac{Cov_{VD,1}(Z_i, Z_j)}{\sigma^2} = \frac{\rho_{\text{data}}}{d}.
    \end{equation}
    
    For $d > 1$, we have the following inequality
    \begin{equation}
      \rho_{\text{VDeH}} > \rho_{\text{IKE}_{\text{VD}}(m=1)}.  
    \label{eq:rho_left}
    \end{equation}
    
    The randomness in IKE, $\Theta_{\text{IKE}} = (D_\psi, \mathcal{R})$, involves the data sample $D_\psi$ and a rich set of subsequent random choices $\mathcal{R}$ to construct the entire iTree (i.e., random dimension selections and value splits). 
    This randomness further reduces the correlation $\rho$ between partitioners. Thus, we have
    \begin{equation}
    \rho_{\text{IKE}_{\text{VD}}(m=1)} > \rho_{\text{IKE}} \ge 0.
    \label{eq:rho_right}
    \end{equation}
    Combing Eq.(\ref{eq:rho_left}) and Eq.(\ref{eq:rho_right}) completes the proof.
\end{proof}

Theorem~\ref{thm:rho} shows that the randomization process employed by IKE makes it superior in generating diverse base partitioners, leading to a robust binary hash representation. 

\section{Method Details}\label{app:method_details}
This section provides details of our method. The code is publicly available at: \url{https://github.com/Zed-Zed-b/IK_Emb}.

\subsection{Construction Details of Isolation Tree (iTree)}\label{app:method_details:iTree}
Given a dataset \(\mathcal{D} = \{\mathbf{x}_1,\dots,\mathbf{x}_\psi\} \subset \mathbb{R}^d\), one first chooses a feature \(q\) uniformly at random from the set of all features of \(\mathcal{D}\). Next, a split value \(p\) is drawn uniformly from the interval \([\min_q(\mathcal{D}),\max_q(\mathcal{D})]\), where \(\min_q(\mathcal{D})\) and \(\max_q(\mathcal{D})\) denote the minimum and maximum values of feature \(q\) over all points in \(\mathcal{D}\). The dataset \(\mathcal{D}\) is then partitioned into two disjoint subsets: \(\mathcal{D}_{\ell} = \{\mathbf{x}\in \mathcal{D} \mid x_q < p\}\), \(\mathcal{D}_{r} = \{\mathbf{x}\in \mathcal{D} \mid x_q \ge p\}\).
This splitting process is applied recursively to \(\mathcal{D}_{\ell}\) and \(\mathcal{D}_{r}\), with the current tree height \(e\) incremented by one at each recursive call. Recursion terminates when any of the following conditions is satisfied: (i) \(\lvert \mathcal{D} \rvert = 1\), (ii) the current tree height \(e\) reaches the specified height limit \(l\), or (iii) the chosen split \((q,p)\) fails to produce two nonempty subsets (i.e., \(\mathcal{D}_{\ell} = \varnothing\) or \(\mathcal{D}_{r} = \varnothing\)). The partitioning in iTree is demonstrated in Figure \ref{fig:iforest_partition}. It is noteworthy that the entire construction process of iTree does not require learning. 

\subsection{Formal Description for \modelnamevd{}} \label{app:method_details:ike_vd}
Let $D=[x_1,\dots,x_N]\in\mathbb{R}^{d\times N}$ be the input dataset. An \modelnamevd{} partition is constructed via a three-step process: (1) Randomly selecting \(m\) dimensions (\(m \le d\)); (2) Randomly sampling $\psi$ anchor points $\mathcal{D}=\{s_1,\dots,s_{\psi}\}\subset D$; (3) Constructing a Voronoi Diagram $H_J$ using the only selected dimensions of anchor points. For any point $x\in\mathbb{R}^{d}$, its transformed value is determined by which of the $\psi$ Voronoi cells it falls into (see the demonstration in Figure \ref{fig:VD_partition}).

\subsection{An Example for Efficient Similarity Computation}\label{app:method_details:example}

This section provides a step-by-step example illustrating the efficient computation of matching elements between two bit strings. Let \(D_x\) and \(D_y\) be two 8-bit strings, each segmented into \(n_b = 2\) bit element. Thus, each string consists of \(8/2=4\) elements. The goal is to count the number of positions where the corresponding elements match exactly.

Assume the values of \(D_x\) and \(D_y\) are:
\begin{align*}
    D_x &= 00\;01\;10\;01 \\
    D_y &= 00\;10\;11\;01 \\
\end{align*}

The bitwise XOR operation identifies positions where the two bits differ. A segment of all zeros after XOR indicates a matching element.
\[
    D=D_x \oplus D_y = 00\; 11\; 01\; 00
\]


Checking each segment for all zeros directly is inefficient. Instead, a logical right shift followed by a bitwise OR is used to propagate any '1' within a segment to its rightmost bit. If the rightmost bit is '0' after this operation, it indicates that all bits in this segment are 0 (meaning the element matches):
\begin{align*}
    M & = D \; | \;(D \gg 1) \\
    &= 00\; 11\; 01\; 00 \; | \; 00\;01\;10\;10 \\
    &= 0\textbf{0} \; 1\textbf{1} \; 1\textbf{1} \; 1\textbf{0}
\end{align*}
The bold bits represent the rightmost bit of each segment.

To count segments where the rightmost bit is '0', all other bits (non-rightmosts) in \(M\) are set to '1' by OR-ing with a pre-defined mask. The mask for \(n_b=2\) is 10 10 10 10. For \(n_b=4\), it would be 1110 1110.

\begin{align*}
M &= M \; | \; mask \\    
&=0\textbf{0} \; 1\textbf{1} \; 0\textbf{1} \; 1\textbf{0} \; | \; 10\;10\;10\;10\\
&=1\textbf{0} \; 1\textbf{1} \; 1\textbf{1} \; 1\textbf{0}
\end{align*}

The number of ‘1’ bits in \(M\) is counted using the population count (popcnt) function. The number of matching elements equals the length of \(M\) minus this count, 
\begin{align*}
    n_1 &= \text{popcnt}(M) = 6 \\
    n_{match} &= \text{length}(M) - n_1 = 2
\end{align*}
Thus, between \(D_x\) and \(D_y\), two out of four elements match.

\subsection{Discussion on Database Scalability and Distribution Shifts}

\modelname{} natively supports efficient incremental updates without requiring a full system rebuild. During the initialization phase, \modelname{} utilizes a representative sub-sample set to construct a fixed mapping function $f(x): \mathbb{R}^d \to \{0, 1\}^{t \times \psi}$ based on the iTree structure. Once this mapping function is determined, it functions as a static, data-dependent hashing mechanism. For newly arriving data points, the system only needs to perform a single forward pass to convert the high-dimensional vectors into binary codes using the existing iTrees. These codes are then directly inserted into the existing indexing structures (e.g., HNSW or inverted lists). This property ensures that \modelname{} can seamlessly scale to expanding databases with minimal computational overhead.

A common challenge in industrial-scale retrieval is the evolution of document corpora over time. While the mapping function remains effective as long as the underlying data distribution is relatively stable, a significant distribution shift may necessitate an update to the model to maintain optimal retrieval precision. Unlike learning-based quantization or hashing methods that require expensive retraining on massive datasets, \modelname{} maintains a distinct advantage due to its learning-free nature and extremely low reconstruction cost. In the event of a detected distribution shift, \modelname{} can rapidly reconstruct the mapping function through efficient re-sampling.

To quantify this efficiency, we conducted benchmarks on the Istella22 dataset (1 million data points) using the LLM2Vec model ($t = 4096, \psi = 16$). On the Linux server reported in the Section \ref{sec:Experiments}, a complete round of sample re-sampling and model reconstruction takes only 0.045 s. Furthermore, remapping the entire 1-million-sample corpus with the updated model is completed within 3.21 s. Such high responsiveness makes \modelname{} a highly cost-effective solution for dynamic, industrial-scale retrieval scenarios compared to traditional learning-based approaches.

\section{Experiment Details}
In this section, we provide the details of the datasets and hyperparameter tuning procedure employed in our study.

\begin{figure*}[!ht]
  \centering
  \begin{minipage}{\linewidth}
    \centering
    
    \subfloat[][Istella22 (LLM2Vec)]{\label{app:fig:flat_ann_istella22_l2v}\includegraphics[width=.37\linewidth]{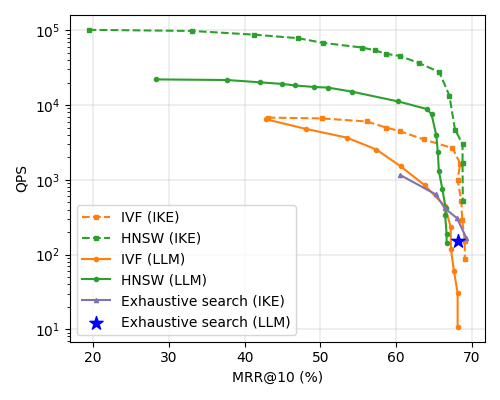}}\hspace{0.1\linewidth}
    \subfloat[][TREC DL 23 (LLM2Vec)]{\label{app:fig:flat_ann_trecdl23_l2v}\includegraphics[width=.37\linewidth]{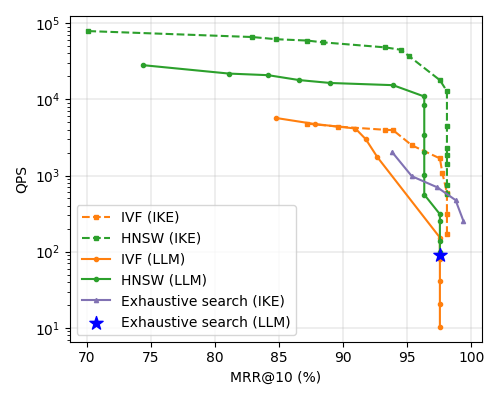}}

  \end{minipage}
  
  \begin{minipage}{\linewidth}
    \centering
    
    \subfloat[][Istella22 (Qwen3)]{\label{app:fig:flat_ann_istella22_qwen3}\includegraphics[width=.37\linewidth]{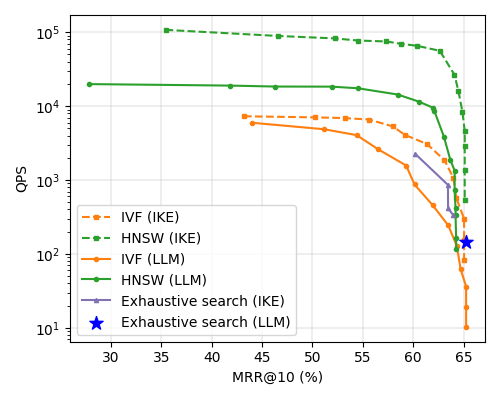}}\hspace{0.1\linewidth}
    \subfloat[][TREC DL 23 (Qwen3)]{\label{app:fig:flat_ann_trecdl23_qwen3}\includegraphics[width=.37\linewidth]{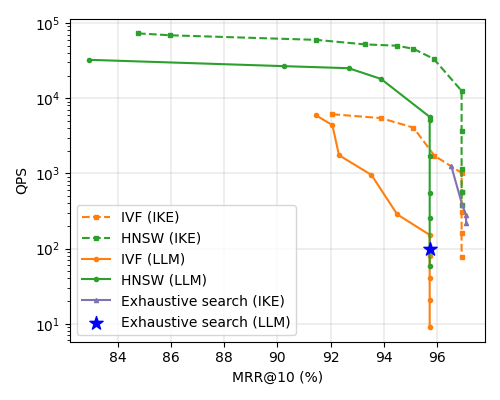}}

  \end{minipage}
  
  \caption{QPS vs. MRR@10 for Four ANN Methods on Istella22 and TREC DL 23.}
  \label{app:exp:fig:ann_flat}
\end{figure*}

\subsection{Dataset Details} \label{app:exp_detail:dataset}
The Massive Text Embedding Benchmark (MTEB)~\cite{muennighoff2022mteb} is a widely used benchmark for validating LLMs' ability to extract text embeddings. We exploit 3 retrieval datasets from the MTEB (eng, v2) benchmark: FiQA2018~\cite{thakur2021beir}, FEVER-HN~\cite{fever}, and Touche2020.V3~\cite{touchev3}. To assess scalability with large corpora, we also include an additional dataset from the MTEB (eng, v1) benchmark: HotpotQA~\cite{hotpotqa}, which has a corpus size of 5.2 million. To further validate the effectiveness of our method on multi-level retrieval datasets, we use the Istella22~\cite{istella22} and TREC DL 23~\cite{trecdl23} datasets (with a 1 million subset extracted from their original corpora). All datasets provide human-annotated relevance labels for each query. In multi-level datasets like Istella22 and TREC DL, documents are rated on a graded relevance scale, ranging from 1 to 5 in Istella22 and from 1 to 4 in TREC DL. Detailed dataset statistics are presented in Table \ref{app:dataset_tat}.

\begin{table}[!ht]
    \centering
    \caption{\textbf{Dataset Statistics}. Datasets are categorized by source. ``HN'' in the dataset name indicates ``Hard Negatives''. Relevance indicates the relationship between query and document: binary score (relevant, non-relevant) or multi-level score, in which each document is assigned the level of relevance w.r.t the query.}
    \label{app:dataset_tat}
    \resizebox{\columnwidth}{!}{%
    \begin{tabular}{@{}cccc@{}}
        \toprule
        \textbf{Dataset} & \textbf{Relevance} & \textbf{Corpus Size} & \textbf{Query Size} \\ 
        \midrule
        \multicolumn{4}{@{}c}{\textit{MTEB}} \\
        \cmidrule(lr){1-4}  

        HotpotQA    & Binary & 5,233,329 & 7,405  \\
        FiQA2018      & Binary & 57,638 & 648 \\
        FEVER-HN & Binary & 163,698 & 1,000 \\
        Touche2020.V3 & 3-Level & 303,732 & 49 \\
        
        \cmidrule(lr){1-4}
        \multicolumn{4}{@{}c}{\textit{Extra Datasets}} \\
        \cmidrule(lr){1-4}  
        
        Istella22 (subset) & 5-Level & 1,000,000 & 234 \\
        TREC DL 23 (subset) & 4-Level & 1,000,000 & 82 \\
        
        \bottomrule
    \end{tabular}
    }
\end{table}

\subsection{Hyperparameter Tuning} \label{app:exp_detail:hyperpara}
The parameter \(\psi\), which controls the number of partitions, is determined for each dataset through a grid search on a validation set, aimed at optimizing retrieval accuracy under an exhaustive search setting. This process is efficient, requiring only a single grid search per dataset. Specifically, we tune over the range \([2,16]\) and select the best value.

Notably, the FEVER-HN and Touche2020.V3 datasets do not provide an official validation split. We construct their validation sets as follows:

\begin{itemize}
\item \textbf{FEVER-HN}: This dataset is a subset derived from the FEVER dataset~\cite{fever}. We therefore adopt the official validation split of the original FEVER dataset as the validation set for FEVER-HN. The official test set of FEVER-HN is subsequently used for final evaluation.
\item \textbf{Touche2020.V3}: We randomly sample 30\% of the official test set to serve as the validation set, using the remaining 70\% for final testing. 
\end{itemize}

\section{Extension of section \ref{sec:comparsion_llm}}

\subsection{Implementation details of combining two ANN methods with IKE}\label{app:exp:llm_emb:ann_details}

\paragraph{IVF ({\modelname})} We use \(k\)-means clustering to partition the data and construct the IVF index in the LLM embedding space. During retrieval, the query's LLM embedding is used to identify the closest clusters, and then the similarity between the query and the data points within these clusters is computed in the {\modelname} space.
\paragraph{HNSW ({\modelname})} The index construction and search processes follow the methodology proposed in the original paper, with the only modification being the replacement of the cosine similarity defined in the LLM embedding space with the similarity defined in the {\modelname} space.

\subsection{Detailed Experimental Configuration}\label{app:exp:llm_emb:config}
For the two ANN algorithms with  {\modelname}, we retain the same value of \(\psi\) in the exhaustive search setting for each dataset. For the IVF index, the number of clusters \(nlist\) is set to 4,096, following the recommendation in Faiss~\cite{douze2024faiss} for million-scale datasets. The hyperparameters for the HNSW index are set as follows: $M=32$ and $efConstruction=500$. In each plot, dotted lines represent retrieval performed in the {\modelname} space, whereas solid lines represent retrieval in the LLM embedding space. The same color indicates the same underlying index structure for direct comparison. Additionally, the blue pentagram markers denote the retrieval performance achieved through exhaustive search in the LLM embedding space. The purple curve corresponds to the retrieval performance of {\modelname} under the exhaustive search setting, which was obtained by tuning the hyperparameter \(t\) while keeping the parameter \(\psi\) fixed.

\subsection{Additional experimental results in the ANN search setting}\label{app:exp:llm_emb:more_exp_res}
Additional experimental results for the ANN search setting are provided in Figure \ref{app:exp:fig:ann_flat}, which shows the QPS-MRR@10 trade-off curves on the Istella22 and TREC DL 23 datasets. 
Consistent with the results presented in the main paper, performing retrieval with IKE using the HNSW index achieves multiple times higher throughput compared to retrieval directly with LLM embeddings. For instance, on the Istella22 (Qwen3) dataset under the around 60\% MRR@10 condition, QPS improves from \(1.1\times10^4\) (LLM) to \(6.5\times10^4\) ({\modelname}), representing an nearly 6\(\times\) acceleration in retrieval speed. Additionally, on some datasets, the peak performance on {\modelname} exceeds that of the LLM, which aligns with the findings presented in~\autoref{tab:comp_with_llm}. These results demonstrate that retrieval in the {\modelname} space provides substantial efficiency gains while achieving comparable or even better accuracy.


\section{Extension of section \ref{sec:learningBased}}\label{app:exp:learn_based}


Figure \ref{app:exp:fig:CSR_IKE} presents the additional experimental results in the exhaustive search setting between CSR and {\modelname} on the Istella22 (LLM2Vec) and TREC DL 23 (LLM2Vec) datasets. The results show that {\modelname} significantly outperforms CSR in terms of search efficiency across all datasets, while achieving comparable or even superior retrieval accuracy.

\begin{figure*}[]
  \centering
  \begin{minipage}{\linewidth}
    \centering
    
    \subfloat[][Istella22 (LLM2Vec)]{\label{app:fig:CSR_istella22_l2v}\includegraphics[width=.40\linewidth]{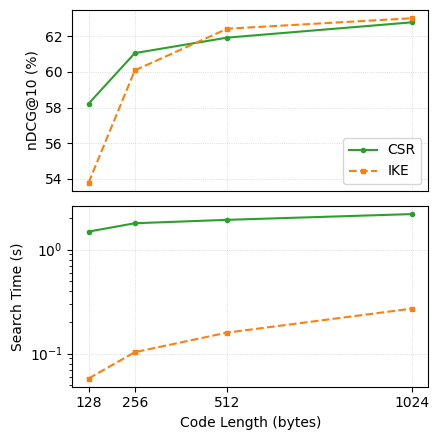}}\hspace{0.1\linewidth}
    \subfloat[][TREC DL 23 (LLM2Vec)]{\label{app:fig:CSR_trecdl23_l2v}\includegraphics[width=.40\linewidth]{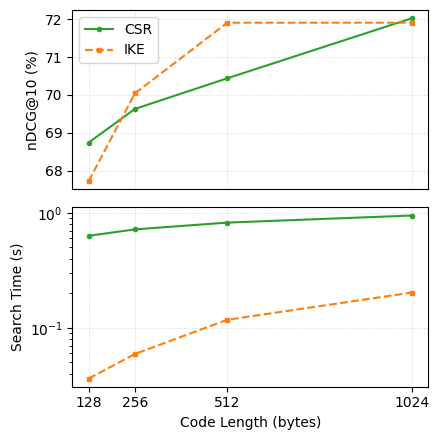}}

  \end{minipage}
  
  \caption{\textbf{Comparison of Retrieval Accuracy and Search Time between CSR and {\modelname} in the exhaustive search setting on the Istella22 (LLM2Vec) and TREC DL 23 (LLM2Vec) datasets}. The search time includes both the mapping time for all queries and the similarity computation time, rather than solely the similarity calculation time as reported in the original CSR paper.}
  \label{app:exp:fig:CSR_IKE}
\end{figure*}

\section{Extension of section \ref{sec:learningFree_with_index}}
\subsection{Description of Compared ANN methods}
\label{app:exp:learn_free:methods}
The compared ANN methods are described in detail as follows:
\begin{itemize}
    \item \textbf{IVF ({\modelname}) / HNSW ({\modelname})}: This refers to two primary ANN methods, operated in the {\modelname} space. The detailed implementations of these methods are the same as those in the previous section.

    \item \textbf{IVF (rpLSH) / HNSW (rpLSH)}: rpLSH is a short name for Random-Projection-based Locality Sensitive Hashing~\cite{rpLSH}, which is a simple yet powerful LSH algorithm. Previous research has demonstrated its superior effectiveness, ranking it first among eleven popular hashing algorithms~\cite{revise_LSH}. It should be noted that Faiss only provides an exhaustive search implementation of rpLSH. To ensure a fair comparison, we integrate rpLSH with the IVF and HNSW indexing structures to perform an ANN search.

    \item \textbf{IVF (PQfs)}: Product Quantization (PQ) \cite{PQ} is widely adopted in practice due to its effective balance between retrieval accuracy and compression capability \cite{diskann-github, johnson2019billion, PQcache, qiu-etal-2022-efficient}. However, conventional PQ implementations can suffer from slower search speeds due to table lookup operations. A recent approach, Product Quantization Fast Scan (PQfs)~\cite{PQfs}, utilizes a SIMD-based fast implementation that significantly accelerates the process, achieving a 4 to 6 times speedup while maintaining the same level of accuracy as standard PQ. IVF (PQfs) integrates PQfs with the IVF indexing structure. We utilize the Faiss implementation of IVF (PQfs) as a baseline for comparison. 

    \item \textbf{ScaNN}: ScaNN~\cite{scann}, developed by Google, employs a novel anisotropic quantization loss function for constructing the quantization codebook. We use the official ScaNN implementation in its ANN search mode for evaluation.
    
\end{itemize}

\subsection{Additional experimental results in the ANN search setting}\label{app:exp:learn_free:more_exp_res}

\sloppy
Figure \ref{app:exp:fig:learn_free_ann} illustrates the QPS versus nDCG@10 trade-off curves for all methods on the HotpotQA (Qwen3), Istella22, and TREC DL 23 datasets, with the code lengths used indicated above each plot. In terms of search efficiency, HNSW (\modelname{}) achieves up to 10\(\times\) higher throughput than other learning-free compression methods in ANN setting, while maintaining comparable retrieval accuracy. For instance, on the Istella22 (Qwen3) dataset with a 512-byte code length, HNSW ({\modelname}) improves throughput from \(1.1\times10^4\) to \(1.5\times10^5\) QPS at around 51\% nDCG@10 compared to ScaNN, resulting in an order-of-magnitude improvement. Furthermore, the IVF ({\modelname}) method attains search efficiency on par with the IVF (PQfs) approach. 


Regarding accuracy, all methods achieve similar top-line performance across most datasets. Notably, our method attains significantly higher retrieval accuracy on the Istella22 (Qwen3) dataset, as shown in Figure \ref{app:fig:compression_istella22_qwen3}, where it outperforms other methods by 3\% at a code length of 512 bytes. It is also noteworthy that the rpLSH method achieves competitive accuracy on several datasets. However, its retrieval speed is substantially lower than that of our approach, primarily due to its embedding mapping process involving high-dimensional matrix multiplications with high computational overhead. As the code length increases, the QPS gap between HNSW (rpLSH) and HNSW ({\modelname}) gradually becomes large. This occurs because the mapping time begins to dominate the total retrieval process of rpLSH, thereby diminishing the efficiency advantage typically offered by the HNSW index.

\begin{figure*}[htbp]
  \centering
  

  \begin{subfigure}{\textwidth}
    \centering
    \includegraphics[width=0.96\linewidth]{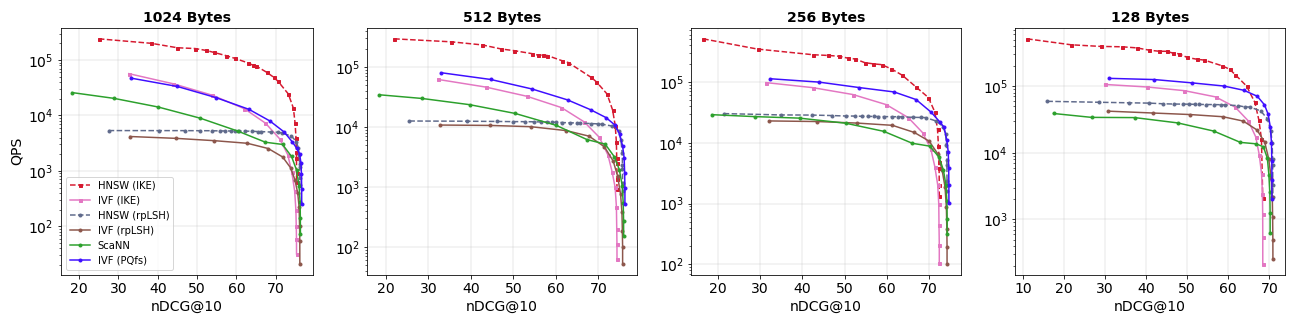}
    \caption{HotpotQA (Qwen3)}
    \label{app:fig:compression_hotpotqa_qwen3}
  \end{subfigure}
  
  \begin{subfigure}{\textwidth}
    \centering
    \includegraphics[width=0.96\linewidth]{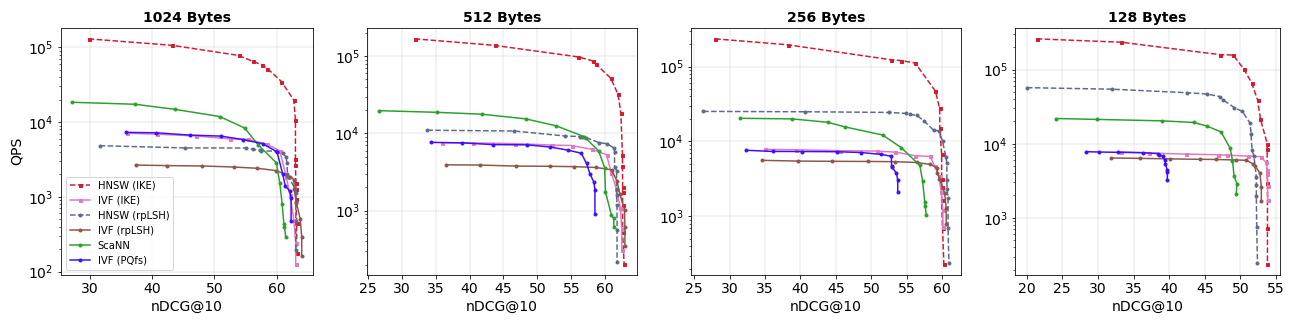}
    \caption{Istella22 (LLM2Vec)}
    \label{app:fig:compression_istella22_l2v}
  \end{subfigure}

  \begin{subfigure}{\textwidth}
    \centering
    \includegraphics[width=0.96\linewidth]{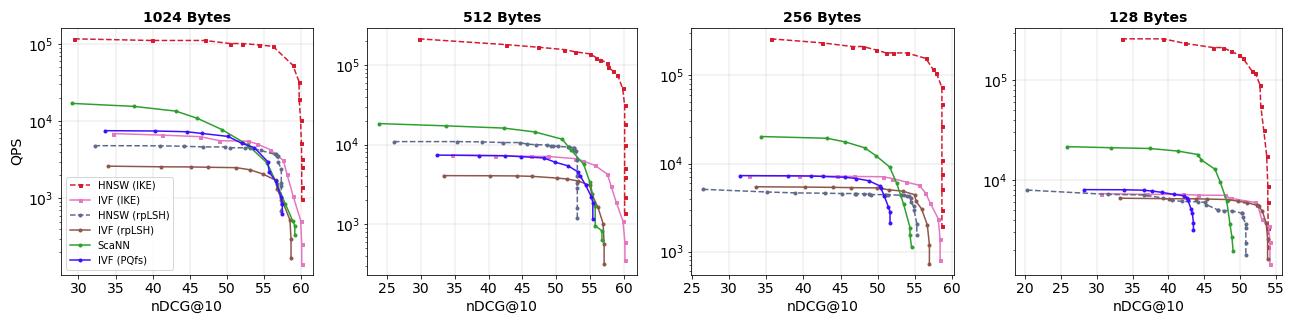}
    \caption{Istella22 (Qwen3)}
    \label{app:fig:compression_istella22_qwen3}
  \end{subfigure}

  \begin{subfigure}{\textwidth}
    \centering
    \includegraphics[width=0.96\linewidth]{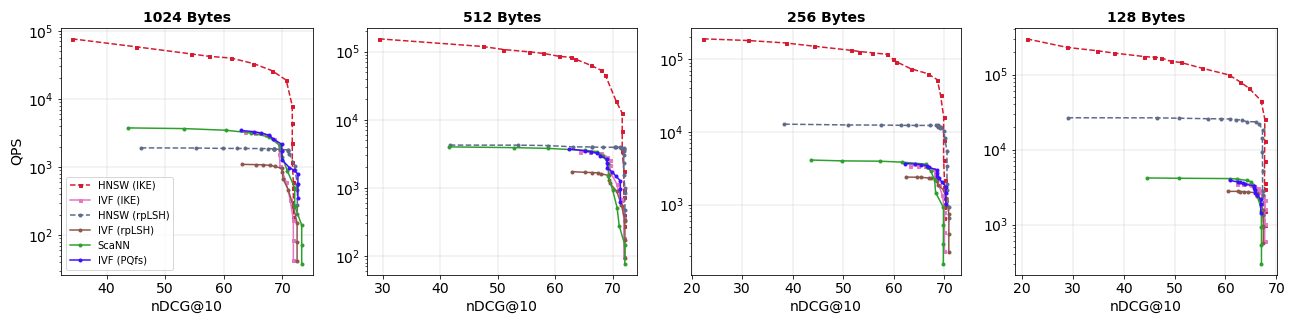}
    \caption{TREC DL 23 (LLM2Vec)}
    \label{app:fig:compression_trecdl23_l2v}
  \end{subfigure}

  \begin{subfigure}{\textwidth}
    \centering
    \includegraphics[width=0.96\linewidth]{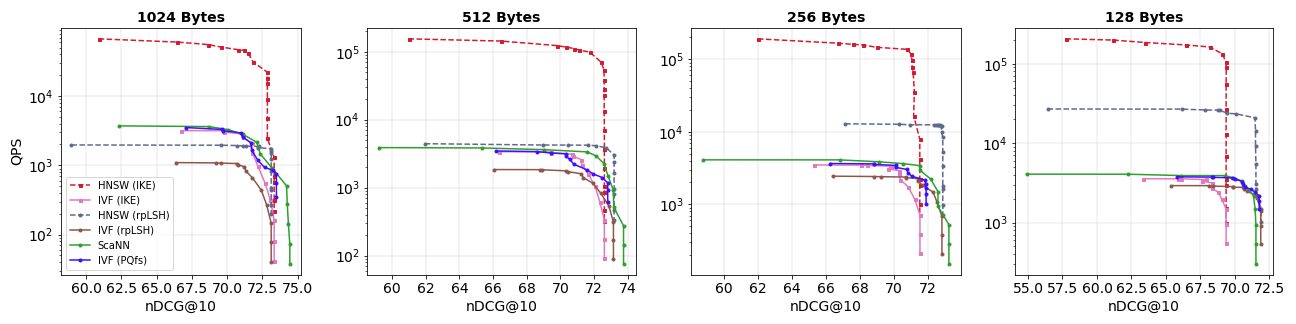}
    \caption{TREC DL 23 (Qwen3)}
    \label{app:fig:compression_trecdl23_qwen3}
  \end{subfigure}

  \caption{\textbf{Performance Comparison of Compression Methods on ANN Retrieval: QPS vs. nDCG@10 on HotpotQA (Qwen3), Istella22, and TREC DL 23}. The curves are obtained by adjusting the search parameter of the index structures.}
    \label{app:exp:fig:learn_free_ann}
\end{figure*}

\end{document}